\documentclass[10pt,journal,compsoc]{IEEEtran}

\ifCLASSOPTIONcompsoc
  \usepackage[nocompress]{cite}
\else
  \usepackage{cite}
\fi

\usepackage{cite,graphicx,amsmath,amssymb,mathrsfs,epsf}
\usepackage[usenames]{color}
\usepackage{multirow}
\usepackage{tabularx}
\usepackage{stfloats}
\usepackage{setspace}
\usepackage{bbold}
\usepackage{bm}
\usepackage{bigstrut,array,multirow,tabularx}
\usepackage{tikz}
\usepackage{xcolor}
\usepackage{pgfplots}
\usepackage{dsfont}
\usepackage{bbm}
\usepackage{graphics}
\usepackage{graphicx}

\usepackage[hidelinks]{hyperref}
\usepackage{color}
\usepackage{multirow}
\usepackage{tabularx}
\usepackage{adjustbox}
\usepackage{fixltx2e}
\usepackage{caption}

\usepackage{multicol}

\usepackage{algorithm}
\usepackage[noend]{algpseudocode}
\algrenewcommand\algorithmicrequire{\textbf{Input:}}
\algrenewcommand\algorithmicensure{\textbf{Output:}}
\usepackage[caption=false,font=normalsize,labelfont=sf,textfont=sf]{subfig}
\usepackage{cleveref}
\usepackage{verbatim}
\usepackage{bm}
\usepackage{textcomp}
\usepackage{amsthm}

\begin{document}

\title{
	Improved Density-Based Spatio--Textual Clustering on Social Media
}

\author {
       	 Minh~D.~Nguyen and~Won-Yong Shin,~\IEEEmembership{Senior Member,~IEEE}%
\IEEEcompsocitemizethanks{
\IEEEcompsocthanksitem Minh~D.~Nguyen with the Department of Mobile Systems Engineering, Dankook University, Yongin 16890, Republic of Korea. \protect\\
E-mail: minhnguyen2908@gmail.com.  
\IEEEcompsocthanksitem Won-Yong Shin with the Department of Computer Science and Engineering, Dankook University, Yongin 16890, Republic of Korea.\protect\\
E-mail: wyshin@dankook.ac.kr.}}%

\markboth{Submitted to IEEE Transactions on Knowledge and Data Engineering}
{Minh {et al.}: Improved Density-Based Spatio--Textual Clustering on Social Media}

\newtheorem{definition}{Definition}
\newtheorem{theorem}{Theorem}
\newtheorem{lemma}{Lemma}
\newtheorem{example}{Example}
\newtheorem{corollary}{Corollary}
\newtheorem{proposition}{Proposition}
\newtheorem{conjecture}{Conjecture}
\newtheorem{remark}{Remark}

\newcommand{\red}[1]{{\textcolor[rgb]{1,0,0}{#1}}}

\def \diag{\operatornamewithlimits{diag}}
\def \log{\operatorname{log}}
\def \rank{\operatorname{rank}}
\def \out{\operatorname{out}}
\def \exp{\operatorname{exp}}
\def \arg{\operatorname{arg}}
\def \E{\operatorname{E}}
\def \tr{\operatorname{tr}}
\def \SNR{\operatorname{SNR}}
\def \dB{\operatorname{dB}}
\def \ln{\operatorname{ln}}

\def \be {\begin{eqnarray}}
\def \ee {\end{eqnarray}}
\def \ben {\begin{eqnarray*}}
\def \een {\end{eqnarray*}}

\newcommand{\Pro}[1]{\mathrm{Pr}\left\{#1\right\}}
\newcommand{\LIF}[2]{\tilde{L}_{\pi_1(#1),#2}}
\newcommand{\TIL}[2]{L_{\pi_2(#1),#2}}
\newcommand{\TIF}[2]{T_{\pi_1(#1),#2}}
\newcommand{\KIF}[2]{T_{\pi_1(#1),\pi_2(#2)}}
\newcommand{\snr}{\textsf{snr}}
\newcommand{\sinr}{\textsf{sinr}}
\newcommand{\CanSB}{\mathcal{B}}
\newcommand{\CanSA}{\mathcal{A}}
\newcommand{\Norm}[1]{\left|{#1}\right|}

\newcommand{\argmin}{\mathop{\mathrm{arg\,min}}}
\newcommand{\argmax}{\mathop{\mathrm{arg\,max}}}

\newcommand{\N}{\mathbbmss{N}}
\newcommand{\R}{\mathbbmss{R}}
\newcommand{\C}{\mathbbmss{C}}
\newcommand{\Z}{\mathbbmss{Z}}
\newcommand{\B}[1]{\mathbf{#1}}
\newcommand{\myindent}[1]{
\newline\makebox[#1cm]{}
}
\IEEEtitleabstractindextext{
\begin{abstract}
DBSCAN may not be sufficient when the input data type is heterogeneous in terms of textual description. When we aim to discover clusters of geo-tagged records relevant to a particular point-of-interest (POI) on social media, examining only one type of input data (e.g., the tweets relevant to a POI) may draw an incomplete picture of clusters due to noisy regions. To overcome this problem, we introduce \textsf{DBSTexC}, a newly defined density-based clustering algorithm using {\em spatio--textual} information. We first characterize {\em POI-relevant} and {\em POI-irrelevant} tweets as the texts that include and do not include a POI name or its semantically coherent variations, respectively. By leveraging the proportion of POI-relevant and POI-irrelevant tweets, the proposed algorithm demonstrates much higher clustering performance than the DBSCAN case in terms of $\mathcal{F}_1$ score and its variants. While \textsf{DBSTexC} performs exactly as DBSCAN with the textually homogeneous inputs, it far outperforms DBSCAN with the textually heterogeneous inputs. Furthermore, to further improve the clustering quality by fully capturing the geographic distribution of tweets, we present fuzzy \textsf{DBSTexC} (\textsf{F-DBSTexC}), an extension of \textsf{DBSTexC}, which incorporates the notion of fuzzy clustering into the \textsf{DBSTexC}. We then demonstrate the robustness of \textsf{F-DBSTexC} via intensive experiments. The computational complexity of our algorithms is also analytically and numerically shown.
\end{abstract}

}
\maketitle

\IEEEdisplaynotcompsoctitleabstractindextext

%
\IEEEpeerreviewmaketitle
\section{Introduction}
%
%
%
%
\subsection{Background}
\IEEEPARstart{C}lustering is one of the prominent tasks in exploratory data mining, and a common technique for statistical data analysis. Cluster analysis refers to the partitioning of objects into a finite set of categories or clusters so that the objects in one cluster have high similarity but are clearly dissimilar to objects in other clusters~\cite{han2011data}. Several different approaches to clustering have extensively been introduced in the literature. For example, algorithms such as K-means~\cite{hartigan1979algorithm} and Clustering Large Applications based on Randomized Search
 (CLARANS)~\cite{ng2002clarans} were designed based on a partitioning approach; Gaussian mixture models~\cite{fraley2002model} and COBWEB~\cite{fisher1987improving} belong to a model-based approach; Divisive Analysis (DIANA)~\cite{kaufman2009finding} and Balanced Iterative Reducing and Clustering using Hierarchies (BIRCH)~\cite{zhang1996birch} were developed based on a hierarchical approach; Statistical Information Grid (STING)~\cite{wang1997sting} and Clustering in Quest (CLIQUE)~\cite{clique} were shown as a grid-based approach; and Density-Based Spatial Clustering of Applications with Noise (DBSCAN)~\cite{ester1996density} and Ordering Points to Identify the Clustering Structure (OPTICS)~\cite{ankerst1999optics} are examples of a density-based approach.

Among those approaches, density-based clustering has been extensively studied to discover insights in geographic data~\cite{sander1998density}. Due to the fact that density-based clustering returns clusters of an arbitrary shape, is robust to noise, and does not require prior knowledge on the number of clusters, it is suitable for diverse nature-inspired applications~\cite{kriegel2011density}. For instance, through density-based clustering on geographic data, researchers are capable of finding clusters of restaurants in a city, clusters along roads and rivers, and so forth. Due to its robust performance and intuitive representation, DBSCAN stands out as the most frequently used density-based clustering algorithm. Variations of DBSCAN were also widely studied in~\cite{sander1998density, birant2007, campello2013density, wu2018, bryant2018}. 

Recently, owing to the popularity of social networks (or social media), the volume of spatio--textual data is rising drastically. Hundreds of millions of users on social media tend to share their geo-tagged media contents such as photos, videos, musics, and texts. For example, when users visit a point-of-interest (POI), they are likely to check in, upload photos of their visit, or post geo-tagged textual data via social media to describe their individual idea, feeling or preference relevant to the POI. As an example, more than five hundred million tweets are posted on Twitter~\cite{kwak2010} everyday,\footnote{www.internetlivestats.com/ accessed on November 9, 2017} and approximately 1\% of them are geo-tagged~\cite{morstatter2013sample}, which correspond to five million geo-tagged tweets everyday. As a result, there is a high demand for processing and making good use of spatio--textual information based on massive datasets of real-world social media. While there were several studies on the spatio-textual queries~\cite{de2008keyword, cong2009efficient, yao2010approximate, tao2014fast}, which are to find objects satisfying certain spatial and textual constraints, researches on spatio-textual data analysis by clustering~\cite{choi2017k, wu2016density} have not been closely carried out.

\subsection{Motivation and Main Contributions}
Our study is motivated by the insight that when we find clusters (or geographic regions) from geo-tagged records related to a certain POI on social media, DBSCAN~\cite{ester1996density} and its several variations~\cite{sander1998density, birant2007, campello2013density, wu2018, bryant2018} may not give good clustering results. This comes from the fact that while the geographic region surrounding a POI generally comprises two types of {\em heterogeneous} geo-tags that include and do not include annotated keywords about the POI (defined as \textit{POI-relevant} and \textit{POI-irrelevant} geo-tags, respectively), DBSCAN uses only POI-relevant geo-tags in the process of finding clusters. Therefore, although clusters found by DBSCAN seem to correctly discover groups of POI-relevant geo-tags on the surface, they also blindly include geographic regions which contain a large number of undesired POI-irrelevant geo-tags, thus leading to a poor clustering quality. Hence, in the case of such a heterogeneous input data type, the methodology of DBSCAN using only POI-relevant geo-tags may not be a complete solution to finding clusters. It is essential to perform clustering based on a textually heterogeneous input, including both POI-relevant and POI-irrelevant geo-tags, in order not only to find highly dense clusters of POI-relevant data points but also to exclude the regions with a large number of POI-irrelevant points. 

To this end, we introduce \textsf{DBSTexC}, a novel spatial clustering algorithm based on \textit{spatio--textual} information on Twitter~\cite{vu2016geosocialbound, shin2015new}.\footnote{Even if our focus is on analyzing tweets, the dataset on other social media (or micro-blogs) can also be directly applicable to our research.} We first characterize POI-relevant and POI-irrelevant tweets as the texts that include and do not include a POI name or its semantically coherent variations, respectively. By judiciously considering the proportion of both POI-relevant and POI-irrelevant tweets, \textsf{DBSTexC} is shown to greatly improve the clustering quality in terms of $\mathcal{F}_1$ score and its variants including a geographic factor, compared to that of DBSCAN. This gain comes due to the robust ability of \textsf{DBSTexC} that excludes noisy regions which contain a huge number of undesired POI-irrelevant tweets. Note that \textsf{DBSTexC} can be regarded as a generalization of DBSCAN since it performs exactly as DBSCAN with the textually homogeneous inputs and far outperforms DBSCAN with the heterogeneous inputs. 

In a preliminary version~\cite{minh} of this work, we defined the above clustering problem and proposed an effective \textsf{DBSTexC} algorithm. We note that \textsf{DBSTexC} assumes the resulting clusters having {\em strict} boundaries, which however may not fully exploit the entire geographic features of the data. To further improve the clustering quality based on the observation that the geographic distribution of tweets is generally smooth and thus it is not clear which tweets should be grouped as clusters or be treated as noise, we present a fuzzy DBSTexC (\textsf{F-DBSTexC}) algorithm. \textsf{F-DBSTexC} relaxes the contraints on a point's neighborhood density by allowing an ambiguous tweet to belong to a cluster with a distinct membership degree. We empirically evaluate its performance by showing the superiority over the original \textsf{DBSTexC} in terms of our performance metric. This paper subsumes~\cite{minh} by allowing that decision boundaries for clusters can be fuzzy. The runtime complexity of our two algorithms is also analytically shown and our analysis is numerically validated. Our main contributions are five-fold and summarized as follows:
\begin{itemize}
\item We introduce \textsf{DBSTexC}, a new spatial clustering  algorithm, which intelligently integrates the existing DBSCAN algorithm and the heterogeneous textual information to avoid geographic regions with a large number of POI-irrelevant geo-tagged posts in the resulting clusters.
\item We show the evaluation performance of the proposed clustering algorithm in terms of $\mathcal{F}_1$ score and its variants, while demonstrating its superiority over DBSCAN by up to about 60\%.
\item We also present the \textsf{F-DBSTexC} algorithm, an extension of \textsf{DBSTexC}, which incorporates the notion of fuzzy clustering into the \text{DBSTexC} framework, to fully capture the geographic distribution of tweets in various locations.
\item We demonstrate the robust ability of \textsf{F-DBSTexC} that further improves the clustering quality via intensive experiments, compared to that of \textsf{DBSTexC} by up to about 27\% for several POIs that are located especially in sparsely-populated areas. 
\item We analytically and numerically show the computational complexity of our proposed algorithms when two different implementation approaches are employed. 
\end{itemize} 

This paper is the first attempt to integrate the existing DBSCAN and the heterogeneous textual information, and thus our methodology sheds light on how to design highly-improved spatial clustering algorithms by leveraging spatio--textual information on social media.

\subsection{Organization}
The rest of the paper is organized as follows. In Section \ref{prior}, we review the prior work related to our research. Section \ref{data acq} describes how to collect POIs and search for POI-relevant tweets. In Section \ref{method}, we present the proposed \textsf{DBSTexC} algorithm and empirically evaluate its performance. The computational complexity of our algorithm is analytically shown in Section \ref{result}. Section \ref{fuzzy} introduces \textsf{F-DBSTexC}, an extended version of \textsf{DBSTexC}. Finally, Section \ref{summary} summarizes the paper with some concluding remarks.

\subsection{Notations}
The list of all the notations used in our work is presented in Table \ref{notation}. Some notations will be more precisely defined as they appear in later sections of this paper. 

\begin{table}[!t]
\centering
\captionsetup{font=normalsize}
\caption{Summary of notations}
\label{notation}
\begin{tabular}{p{1.2cm}| p{6.8cm}}
\hline
\textbf{Notation} & \textbf{Description} \\ \hline
$\epsilon$                              & Radius of a point's neighborhood                        \\ 
$N_{\text{min}}$                        & Minimum allowable number of POI-relevant tweets in an $\epsilon$-neighborhood of a point\\ 
$N_{\text{max}}$                        & Maximum allowable number of POI-irrelevant tweets in an $\epsilon$-neighborhood of a point\\ 
$\eta$                                  & Precision threshold for a query region    \\ 
$\mathcal{X}$                           & Set of POI-relevant tweets                \\ 
$\mathcal{Y}$                           & Set of POI-irrelevant tweets              \\ 
$\mathcal{X}_{\epsilon}(p)$             & Set of POI-relevant tweets contained in an $\epsilon$-neighborhood of point $p$ \\ 
$\mathcal{Y}_{\epsilon}(p)$             & Set of POI-irrelevant tweets contained in an $\epsilon$-neighborhood of point $p$ \\ 
$\text{dist}(p,q)$                      & Euclidean distance between points $p$ and $q$            \\ 
$C$                                     & A cluster with label $C$      \\ 
$A$                                     & Area of the geographical region covered by clusters       \\ 
$\bar{A}$                           & Normalized area of the geographical region covered by clusters \\ 
$\alpha$                                & Area exponent         \\ 
$\mathcal{F}_1$                         & $\mathcal{F}_1$ score \\ 
$n$                                     & Number of POI-relevant tweets     \\ 
$m$                                     & Number of POI-irrelevant tweets   \\ 
$\mu_p$                                 & Fuzzy score of point $p$   \\ \hline

\end{tabular}
\end{table}

\section{Previous Work}
\label{prior}
Our clustering algorithm is related to four broad areas of research, namely traditional spatial clustering, spatio--textual similarity search, clustering based on spatial and non-spatial attributes, and fuzzy clustering.

\textbf{Spatial clustering.} A variety of spatial clustering algorithms have been developed in the literature. Several algorithms using a partitioning approach were introduced and widely utilized in~\cite{hartigan1979algorithm,ng2002clarans,park2009simple}. Even though such algorithms are useful for finding sphere-shaped clusters, they require prior knowledge on the number of clusters and thus are unable to find clusters of arbitrary shapes. Next, hierarchical clustering algorithms~\cite{kaufman2009finding,zhang1996birch} can be further divided into two types based on the following clustering processes: the agglomerative (bottom-up) process and the divisive (top-down) process. Their strengths lie in the hierarchical relation among clusters and an easy interpretation. However, hierarchical clustering does not have well-defined termination criteria, and if some objects are mis-clustered during the growth of the hierarchy, then such objects will remain in a certain wrong cluster until the clustering process is terminated. In addition, from a density-based point of view, the DBSCAN algorithm~\cite{ester1996density} uses a series of density-connected points to form density-based clusters. Since DBSCAN does not require the number of clusters as an input parameter, and does not assume any underlying probability density behind the clusters, it can discover clusters of arbitrary shapes. As follow-up studies on DBSCAN, numerous algorithms have been developed as follows. GDBSCAN~\cite{sander1998density} generalized DBSCAN by extending the notion of a neighborhood over the traditional $\epsilon$-neighborhood and by using different measures to define the ``cardinality" of the neighborhood; ST-DBSCAN~\cite{birant2007} was designed by discovering clusters based on spatial and temporal attributes; HDBSCAN~\cite{campello2013density} was presented by generating a density-based clustering hierarchy and then extracting a set of significant clusters based on a measure of stability; DCPGS~\cite{wu2018} revised DBSCAN in such as way that places are clustered based on both spatial and social distances between them (i.e., the geo-social network data); and RNN-DBSCAN~\cite{bryant2018} was proposed by defining observation density using reverse nearest neighbors, which leads to a reduction in complexity and is preferable when clusters have high variations in density. Unlike the aforementioned studies, our work aims to integrate the existing DBSCAN and the heterogeneous textual information to avoid noisy regions having numerous POI-irrelevant geo-tags. 

\textbf{Spatio--textual similarity search.} It is of paramount importance to find spatially and textually closest objects to query objects. To offer compelling solutions to this problem, several algorithms~\cite{de2008keyword, cong2009efficient, yao2010approximate, tao2014fast} were introduced. Particularly, a method to answer queries containing a location and a set of keywords was presented in~\cite{de2008keyword}. Next, an indexing framework for processing top-\textit{k} query that considers both spatial proximity and text relevancy was introduced in~\cite{cong2009efficient}. Although these algorithms study the spatio-textual distance between objects, they are inherently different from our proposed approach, which finds density-based spatio--textual clusters using the textually heterogeneous input data type on social media such as Twitter.

\textbf{Clustering based on spatial and non-spatial attributes.} There have been recent studies on the use of spatial and non-spatial attributes to improve the clustering performance in various applications. Spectral clustering was applied in~\cite{van2013community} to identify clusters among gang members based on both the observation of social interactions and the geographic locations of individuals. On the other hand, another clustering method was presented in~\cite{wang2011spatial} to discover clusters that are dense spatially and have high spatial correlation based on their non-spatial attributes. 

\textbf{Fuzzy clustering.} Most of fuzzy clustering algorithms were built upon the fuzzy c-means algorithm~\cite{bezdek1984fcm,miyamoto2008algorithms,li2008}. These algorithms integrate crisp clustering techniques and the theory of fuzzy sets so as to discover clusters whose objects belong to multiple clusters simultaneously with different degrees of membership~\cite{smiti2013soft,zahid2001fuzzy}. However, fuzzy density-based clustering algorithms may or may not allow overlapping clusters. Fuzzy neighborhood DBSCAN (FN-DBSCAN)~\cite{nasibov2009robustness} was proposed by introducing the definition of fuzzy neighborhood size along with various neighborhood membership functions to capture different neighborhood sensitivities. Three extensions of DBSCAN were also presented in~\cite{ienco2016fuzzy}, while producing clusters with distinct fuzzy and overlapping properties. A survey on popular fuzzy density-based clustering algorithms was presented in~\cite{ulutagay2012fuzzy}.

Note that results presented below partially overlap with our prior conference paper~\cite{minh}. The present paper significantly extends the earlier work in several ways, including the proof of correctness of \textsf{DBSTexC} algorithm, the revised analysis of the computational complexity, and the introduction to \textsf{F-DBSTexC}, an extension of \textsf{DBSTexC} that incorporates the notion of fuzzy clustering to capture the entire geographic features of the data.

\section{Data Acquisition and Processing}
\label{data acq}
We first explain how we acquire the Twitter data and choose POIs. Then, for every POI, we outline our approach to searching for POI-relevant and POI-irrelevant tweets.

\subsection{Collecting Twitter Data}
We utilize the Twitter Streaming Application Programming Interface (API)~\cite{twittermonitor}, which is a widely popular tool to collect data from Twitter for various research purposes such as topic modeling, network analysis, and statistical content analysis. Streaming API returns tweets that match a query written by an API user. An interesting finding is that even if Twitter Streaming API returns at most a 1\% sample of all the tweets created at a given moment, it gives an almost complete set of \emph{geo-tagged} tweets despite sampling~\cite{morstatter2013sample}.

The dataset that we use includes a large number of geo-tagged records (i.e., tweets) collected from Twitter users from May 31, 2016 to June 30, 2016 in the UK. We deleted the contents that were generated by the users posting more than three times consecutively at the same exact location, as those were likely to be  products of other services such as Tweetbot, TweetDeck, Twimight, and so forth. Moreover, we notice that each record consists of a number of attributes that can be distinguished by their associated field names. For data analysis, we select the following three attributes from the collected tweets:
\begin{itemize}
\item \emph{text}: actual UTF-8 text of the tweet;
\item \emph{lat}: latitude of the location where the tweet was posted;
\item \emph{lon}: longitude of the location where the tweet was posted.
\end{itemize}

\subsection{Collecting POIs}
We select POIs as popular point locations that people may be interested in and are likely to visit. Moreover, for the geographic diversity, we choose POIs from both populous metropolitan areas and sparsely populated cities. The names of chosen POIs and their geographic regions are shown in Table \ref{tb1}. Based on the UK gridded population dataset~\cite{UKgridded}, we are able to approximate the population as follows: the population density for the areas surrounding POIs in London, Edinburgh, and Oxford is $>$7,000/km\textsuperscript{2}, $<$2,000/km\textsuperscript{2}, and $<$1,000/km\textsuperscript{2}, respectively. 
\begin{center}
\begin{table}[!t]
\centering
\captionsetup{font=normalsize}
\caption{POI names and the corresponding geographic regions}
\label{tb1}
\begin{tabular}{|c|c|}
\hline
\textbf{POI name}    & \textbf{Region} \\ \hline
Hyde Park            & Populous metropolitan area   \\ \hline
Regent's Park        & Populous metropolitan area   \\ \hline
University of Oxford & Sparsely populated city          \\ \hline
Edinburgh Castle     & Sparsely populated city          \\ \hline
\end{tabular}
\end{table}
\end{center}

\subsection{Searching POI-Relevant Tweets}
Since Twitter users tend to convey their interest in a POI by mentioning or tagging it in their tweets, we are able to collect all POI-relevant tweets by querying for keywords related to the POI in the text field of the collected tweets. However, when users type the actual terms of each POI in their tweets, they may misspell or implicitly mention the POI name. We thus implement a keyword-based search for {\em semantically coherent} variations of a POI, which would contain its shortened names, its informal names (if any), and so forth. For a POI formed into a large geographic area, we include names of famous attractions inside the POI to increase the search accuracy. The list of search queries for four POIs shown in Table \ref{tb1} is summarized in Table \ref{POIquery}.\footnote{Search queries for more POIs to be added in Section \ref{fuzzy} are not shown for the sake of brevity.} Therefore, the dataset can be divided into two subgroups of geo-tagged tweets that include and do not include the annotated POI keywords, which correspond to POI-relevant and POI-irrelevant tweets, respectively.
\begin{center}
\begin{table}[!t]
\centering
\captionsetup{font=normalsize}
\caption{POI names and their search queries}
\label{POIquery}
\begin{tabular}{|p{1.5cm}| p{6.5cm}|}
\hline
\textbf{POI name}    & \multicolumn{1}{c|}{\textbf{Search queries}}                                                           \\ \hline
Hyde Park            & \textit{Hyde Park, Kensington Gardens, Royal Park}                                                     \\ \hline
Regent's Park        & \textit{Regent’s Park, London Zoo, tasteoflondon}                                                      \\ \hline
University of Oxford & \textit{Oxford Univ, oxford univ, Univ Oxford}                                                         \\ \hline
Edinburgh Castle     & \textit{Edinburgh Castle, edinburgh castle, EdinburghCastle} \\ \hline
\end{tabular}
\end{table}
\end{center}

\section{Proposed Methodology}
\label{method}
To elaborate on the proposed methodology, we first present the important definitions and analysis that are essential to the design of our algorithm, and show the analysis that validates the correctness of our algorithm. Then, we elaborate on our DBSTexC algorithm.

\subsection{Definitions}
We start by introducing the definition of a query region. A query region is defined as a geographic area from which we collect the geo-tagged tweets for a particular POI. Generally, we are likely to find both POI-relevant and POI-irrelevant tweets inside the region. Nevertheless, since the relevance of information to the POI varies according to the geographic distance between the POI and the locations where the data are generated, tweets posted at locations far away from the POI are likely to have little or no textual description for the POI. We thus focus only on a region that contains almost all relevant tweets but omit the majority of irrelevant tweets that were posted geographically far from the POI, which would lead to a reduced computational complexity. Motivated by this observation, we define a query region as follows: 

\textit{Definition 1 (Query region):}
Given a POI, a query region is a circle whose center corresponds to the center point of the POI's administrative bounding box provided by Google Maps. The radius of the circle is then increased stepwise until Precision of the query region is lower than a threshold $\eta$, where $\eta$ can be set appropriately based on POI types. Here, $\text{Precision}$ of the query region is the ratio of true positives (the number of POI-relevant tweets in the query region) to all predicted positives (the number of all retrieved geo-tagged tweets in the query region).

In Fig. \ref{queryRegion}, we show an example of the query region for Hyde Park. As shown in the figure, starting from the center of the POI, we continue on expanding the query region until the condition in Definition 1 is fulfilled.   

\begin{figure}[!t]
\centering
\captionsetup{font=normalsize}
\includegraphics[width=1\linewidth]{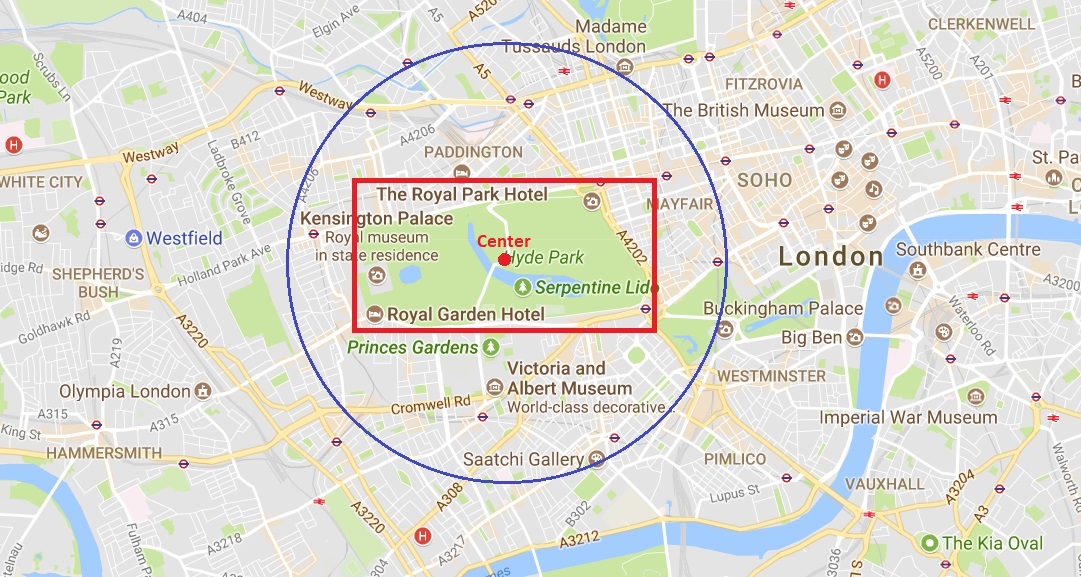}
\caption{An example of the query region for Hyde Park. The red rectangle is the administrative bounding box, whose center is denoted by the red dot, and the blue circle is the query circle that fulfills the condition in Definition 1.}
\label{queryRegion}
\end{figure}

Similar to DBSCAN \cite{ester1996density}, we exploit the neighborhood of a point (See Definition 2) and a series of density-connected points (See Definition 6) to find clusters. However, unlike DBSCAN, we present a new parameter $N_{\text{max}}$ to limit the number of {\em POI-irrelevant} tweets, resulting in an improved clustering quality. Hence, we can acquire a core point which has not only at least $N_{\text{min}}$ POI-relevant tweets but also at most $N_{\text{max}}$ POI-irrelevant tweets inside its neighborhood (See Definition 3). The result of \textsf{DBSTexC}, whose clusters are composed of connected neighborhoods of core points, would significantly outperform DBSCAN that uses only POI-relevant tweets, which is numerically shown in Section \ref{result}. 

\textit{Definition 2 ($\epsilon$-neighborhood of a point):}
Let $\mathcal{X}$ and $\mathcal{Y}$ denote the sets of POI-relevant and POI-irrelevant tweets, respectively. For a point $p\in \mathcal{X}$, the sets of $\epsilon$-neighborhoods containing POI-relevant and POI-irrelevant tweets, denoted by $\mathcal{X}_{\epsilon}(p)$ and $\mathcal{Y}_{\epsilon}(p)$, are defined as the geo-tagged tweets within a scan circle centered at $p$ with radius $\epsilon$ and are given by  
\begin{align*}
\mathcal{X}_{\epsilon}(p) = \{q\in \mathcal{X} | \text{dist}(p,q) \le \epsilon\} \\
\mathcal{Y}_{\epsilon}(p) = \{q\in \mathcal{Y} | \text{dist}(p,q) \le \epsilon\},
\end{align*}
respectively, where dist($p,q$) is the geographic distance between coordinates $p$ and $q$. Note that we focus on the $\epsilon$-neighborhood only for POI-relevant tweets while neglecting the neighborhood of POI-irrelevant tweets, since our \textsf{DBSTexC} algorithm finds clusters based on a series of $\epsilon$-neighborhoods of only POI-relevant tweets. 

\textit{Definition 3 (Core point):} 
A point $p\in \mathcal{X}$ is a core point if it fulfills the following condition:
\begin{equation*}
|\mathcal{X}_{\epsilon}(p)| \ge N_{\text{min}} \; \text{and} \; |\mathcal{Y}_{\epsilon}(p)| \le N_{\text{max}}.
\end{equation*}

\subsection{Analysis}
The analytical part essentially follows the same line as  that in~\cite{clustering}, but is modified so that it fits into our framework. In this subsection, we present fundamental definitions that provide the basis for our \textsf{DBSTexC} algorithm to find clusters according to a density-based approach using spatio--textual information. Then, we analytically validate the correctness of our algorithm by introducing two lemmas.

\textit{Definition 4 (Directly density-reachable):}
A point $p$ is directly density-reachable from a core point $q$ with respect to (w.r.t.) $\epsilon$, $N_{\text{min}}$, and $N_{\text{max}}$ if
\begin{equation*}
p\in \mathcal{X}_{\epsilon}(q) \; \text{or} \; p\in \mathcal{Y}_{\epsilon}(q).
\end{equation*}

If point $p$ is directly density-reachable from a point $q$ and is a core point itself, then $q$ is also directly density-reachable from $p$. Therefore, it is obvious that ``directly density-reachable" is symmetric for pairs of core points.

\textit{Definition 5 (Density-reachable):}
A point $p$ is density-reachable from a point $q$ w.r.t. $\epsilon$, $N_{\text{min}}$, and $N_{\text{max}}$ if there is a chain of points $p_1,\cdots,p_n$, $p_1=q$, and $p_n=p$ such that $p_{i+1}$ is directly density-reachable from $p_i$.

The density-reachable relation is not symmetric. For example, given a directly density-reachable chain as in Definition 5, the points $p_1,\cdots,p_{n-1}$ are all core points. However, $p_n$ can be either a border point or a core point. If $p_n$ is a core point, then point $p_1$ is also symmetrically density-reachable from $p_n$. Therefore, if the two points $p$ and $q$ are density-reachable from each other, then they are core points and belong to the same cluster.  

\textit{Definition 6 (Density-connected):}
A point $p$ is density-connected to a point $q$ w.r.t. $\epsilon$, $N_{\text{min}}$, and $N_{\text{max}}$ if there is a point $o$ such that both $p$ and $q$ are density-reachable from $o$ w.r.t. $\epsilon$, $N_{\text{min}}$, and $N_{\text{max}}$.

With the above six definitions, we are now ready to define a new notion of a cluster. In brief, a cluster (See Definition 7) is defined as a set of density-connected points. Noise points (See Definition 8) are defined as the set of points not belonging to any clusters. 

\textit{Definition 7 (Cluster):} 
Let $\mathcal{T}$ denote the dataset of all retrieved geo-tagged tweets. Then, a cluster $C$ w.r.t. $\epsilon$, $N_{\text{min}}$, and $N_{\text{max}}$ is a non-empty subset of the dataset $\mathcal{T}$ satisfying the following conditions:
\begin{enumerate}
\item \label{imt51} $\forall p \in \mathcal{X}$ and $q \in \mathcal{T}$: if $p \in C$ and $q$ is density-reachable from $p$ w.r.t. $\epsilon$, $N_{\text{min}}$, and $N_{\text{max}}$, then $q \in C$.
\item $\forall p,q \in C$: $p$ is density-connected to $q$ w.r.t. $\epsilon$, $N_{\text{min}}$, and $N_{\text{max}}$.
\end{enumerate}

\textit{Definition 8 (Noise):}
Let $C_1,\cdots,C_k$ be the clusters of the dataset $\mathcal{T}$ w.r.t. $\epsilon_i$, $N_{\text{min}}^i$, and $N_{\text{max}}^i$ for $i\in\{1,\cdots,k\}$. Then, noise is defined as the set of points in $\mathcal{T}$ not belonging to any cluster $C_i$, i.e., $\{p\in \mathcal{T} | p \notin C_i, \forall i\}$.

Given the above eight definitions, our \textsf{DBSTexC} algorithm can then be intuitively stated as a two-step clustering algorithm using spatio--textual information. The first step is to choose an arbitrary POI-relevant tweet satisfying the core point condition as a seed. The second step is to retrieve all points that are density-reachable from the seed, thus obtaining the cluster containing the seed. To formally justify the credibility of our algorithm, we establish the following two lemmas.
\begin{lemma}
\label{lemma1}
Let $p$ be a point in $\mathcal{X}$, $|\mathcal{X}_{\epsilon}(p)|\ge N_{\text{min}}$, and $|\mathcal{Y}_{\epsilon}(p)|\le N_{\text{max}}$. Then, the set $O = \{o|o\in \mathcal{T}$ and $o$ is density-reachable from $p$ w.r.t. $\epsilon$, $N_{\text{min}}$, and $N_{\text{max}}$ $\}$ is a cluster w.r.t. $\epsilon$, $N_{\text{min}}$, and $N_{\text{max}}$.
\end{lemma}
\begin{proof}
Since $p\in \mathcal{X}$, $|\mathcal{X}_{\epsilon}(p)|\ge N_{\text{min}}$ and $|\mathcal{Y}_{\epsilon}(p)|\le N_{\text{max}}$, $p$ is a core point and thus is contained in some cluster $C$. We need to show that $O\subseteq C$. Definition 7-1 indicates that all points that belong to $O$ should also belong to $C$, resulting in $O\subseteq C$. This completes the proof of this lemma.  
\end{proof}
\begin{lemma}
Let $C$ be a cluster w.r.t. $\epsilon$, $N_{\text{min}}$, and $N_{\text{max}}$. Let $p$ be any point in $C\cap \mathcal{X}$ with $|\mathcal{X}_{\epsilon}(p)|\ge N_{\text{min}}$ and $|\mathcal{Y}_{\epsilon}(p)|\le N_{\text{max}}$. Then, $C$ is equal to the set $O = \{o| o$ is density-reachable from $p$ w.r.t. $\epsilon$, $N_{\text{min}}$, and $N_{\text{max}}\}$.
\end{lemma}
\begin{proof}
We need to show that $C = O$. Similarly as in the proof for \textbf{Lemma 1}, we have
\begin{align}
\label{lemma1}
O\subseteq C. 
\end{align}
Therefore, to show that $C = O$, we need to prove that $C\subseteq O$. 
Let $q$ be an arbitrary point in $C$.
Since $p\in C$, $q$ is density-connected to $p$ from Definition 7-2. It means that there is a core point $m\in C$ such that $p$ and $q$ are density-reachable from $m$ (see Definition 6). However, $p$ and $m$ are both core points, which represents that $p$ is density-reachable from $m$ if and only if $m$ is density-reachable from $p$. This shows that $q$ is density-reachable from $p$, which indicates that $q\in O$. Therefore, it follows that
\begin{align}
\label{lemma2}
C\subseteq O. 
\end{align}
From (\ref{lemma1}) and (\ref{lemma2}), we finally have
\begin{equation*}
C = O, 
\end{equation*}
which completes the proof of this lemma.
\end{proof}

\subsection{\textsf{DBSTexC} Algorithm} \label{DBAlgo}
In this subsection, we describe our \textsf{DBSTexC} algorithm that makes use of both POI-relevant and POI-irrelevant tweets. In the clustering process, \textsf{DBSTexC} starts with a random point $p_i$ in $\mathcal{X}$ (i.e., the set of POI-relevant tweets) for $i \in \{1,...,|\mathcal{X}|\}$ and retrieves all points that are density-reachable from $p_i$ with respect to $\epsilon$, $N_{\text{min}}$, and $N_{\text{max}}$ (See Algorithm \ref{alg:Tex1}). If $p_i$ is a core point, then a cluster is formed and expanded until all points that belong to the cluster are included (See Algorithm \ref{alg:Tex2}). Otherwise, \textsf{DBSTexC} moves on to the next point in the set of POI-relevant tweets.
\begin{algorithm}[!t]
\small
\caption{\textsf{DBSTexC}($\mathcal{X}$,$\mathcal{Y}$, $\epsilon$, $N_{\text{min}}$, $N_{\text{max}}$)}
\label{alg:Tex1}
\begin{algorithmic}
\Require $\mathcal{X}$,$\mathcal{Y}$, $\epsilon$, $N_{\text{min}}$, $N_{\text{max}}$
\Ensure Clusters with different labels $C$
\algrenewcommand\algorithmicensure{\textbf{Initialization:}}
\Ensure $C\gets 0$; $n\gets |\mathcal{X}|$; $m\gets |\mathcal{Y}|$; $p_i$ is a point in the set $\mathcal{X}$
\For{each $p_i$}
\If{$p_i$ is not visited}
\State Mark $p_i$ as visited
\State $[\mathcal{X}_\epsilon (p_i), \mathcal{Y}_\epsilon (p_i)] =$ RangeQuery$(p_i)$
\If{$|\mathcal{X}_\epsilon (p_i)|\ge N_{\text{min}} \  \& \  |\mathcal{Y}_\epsilon (p_i)|\le N_{\text{max}}$}
\State $C\gets C+1$
\State ExpandCluster$(p_i,\mathcal{X}_\epsilon (p_i),\mathcal{Y}_\epsilon (p_i))$
\EndIf
\EndIf
\EndFor 
\end{algorithmic}
\end{algorithm}

In Algorithm \ref{alg:Tex1}, RangeQuery() is a function that returns points in an $\epsilon$-neighborhood, where it can be implemented using spatial access methods, i.e., \emph{R-trees} and \emph{k-d trees}. By searching for both POI-relevant and POI-irrelevant points along with two parameters $N_{\text{min}}$ and $N_{\text{max}}$ to determine whether to create a new cluster and/or expand the current cluster, our proposed algorithm effectively excludes noisy areas from its clusters.
\begin{algorithm}[!h]
\small
\caption{ExpandCluster$(p_i, \mathcal{X}_\epsilon (p_i), \mathcal{Y}_\epsilon (p_i))$}
\label{alg:Tex2}
\begin{algorithmic}
\Require $p_i, \mathcal{X}_\epsilon (p_i), \mathcal{Y}_\epsilon (p_i)$
\Ensure Cluster $C$ with all of its members
\algrenewcommand\algorithmicensure{\textbf{Initialization:}}
\State Add $p_i$ to the current cluster
\For{each point $p_j$ in the set $\mathcal{X}_\epsilon (p_i)$}
\If {$p_j$ is not visited}
\State Mark $p_j$ as visited
\State $[\mathcal{X}_\epsilon (p_j), \mathcal{Y}_\epsilon (p_j)] =$ RangeQuery$(p_j)$
\If{$|\mathcal{X}_\epsilon (p_j)|\ge N_{\text{min}}\  \& \  |\mathcal{Y}_\epsilon (p_j)|\le N_{\text{max}}$}
\State $\mathcal{X}_\epsilon (p_i) = \mathcal{X}_\epsilon (p_i) \cup \mathcal{X}_\epsilon (p_j)$
\State $\mathcal{Y}_\epsilon (p_i) = \mathcal{Y}_\epsilon (p_i) \cup \mathcal{Y}_\epsilon (p_j)$
\EndIf
\EndIf
\If{$p_j$ does not have a label}
\State Add $p_j$ to the current cluster
\EndIf
\EndFor
\If{$|\mathcal{Y}_\epsilon (p_i)|\ne 0$}
\For{each point $q_j$ in the set $\mathcal{Y}_\epsilon (p_i)$}
\If{$q_j$ is not visited}
\State Mark $q_j$ as visited
\If{$q_j$ does not have a label}
\State Add $q_j$ to the current cluster
\EndIf
\EndIf
\EndFor
\EndIf 
\end{algorithmic}

\end{algorithm}

In Algorithm \ref{alg:Tex2}, for every point $p_j \in \mathcal{X}_\epsilon (p_i)$, we explore the $\epsilon$-neighborhood of $p_j$. If $p_j$ is a core point, then $p_j$ is added to the current cluster and the algorithm continues by appending its neighbors to the neighbor sets $\mathcal{X}_\epsilon (p_i)$ and $\mathcal{Y}_\epsilon (p_i)$. We repeat this process until all the points in the set $\mathcal{X}_\epsilon (p_i)$ are examined. Eventually, when the process is terminated, the points in the set $\mathcal{Y}_\epsilon (p_i)$ are included in our current cluster. 

\section{Experimental Results and Discussion}
\label{result}
In this section, to show performance of the proposed \textsf{DBSTexC} algorithm in Section \ref{DBAlgo}, we present our performance metric, illustrate experimental results, and analyze the overall average computational complexity.

\subsection{Performance Metric} \label{OpProb}
We choose the $\mathcal{F}_1$ score as a key component of our performance metric, since it is a popular measure in machine learning and statistical analysis for a test's accuracy and thus can be a useful tool to assess the clustering quality. The $\mathcal{F}_1$ score is expressed as
\begin{equation*}
\mathcal{F}_1 = 2\cdot \frac{\text{Precision}\cdot \text{Recall}}{\text{Precision} + \text{Recall}},
\end{equation*}   
which is the harmonic mean of $\text{Precision}$ and $\text{Recall}$. In our work, $\text{Precision}$ is the ratio of true positives (the number of POI-relevant tweets in clusters) to all predicted positives (the number of all geo-tagged tweets in clusters), that is, $\frac{\text{True Positives (TP)}}{\text{TP} + \text{False Positives (FP)}}$; and $\text{Recall}$ is the ratio of true positives to actual positives that should have been returned (the total number of POI-relevant tweets), that is, $\frac{\text{TP}}{\text{TP}+\text{False Negatives (FN)}}$. 

\begin{figure}[!t]
\captionsetup[subfigure]{font=scriptsize,labelfont=normalsize}
\subfloat[Hyde Park\label{sfig:testa}]{%
  \includegraphics[height=3cm,width=.49\linewidth]{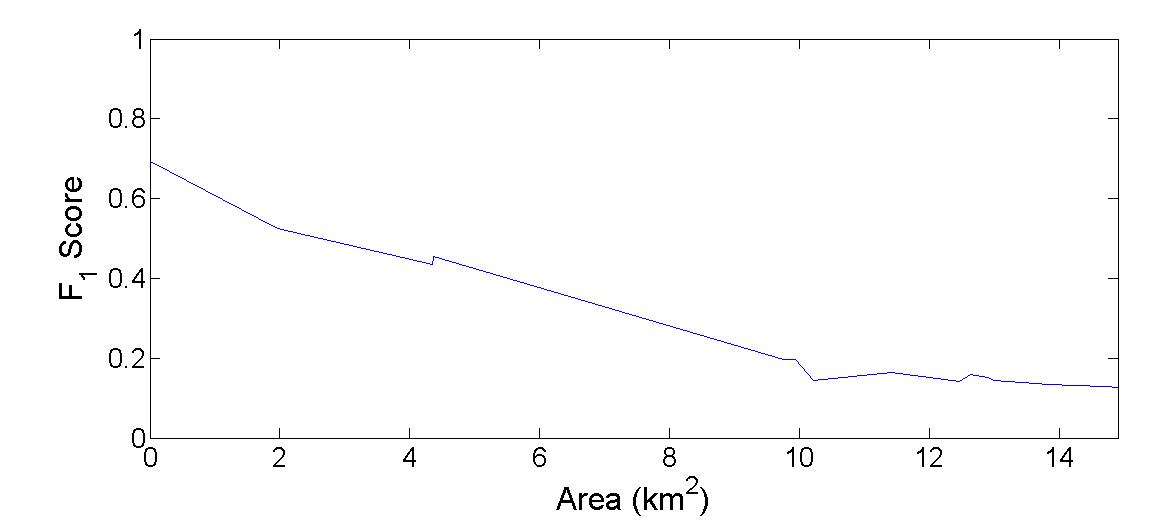}%
}\hfill
\subfloat[Regent's Park\label{sfig:testa}]{%
  \includegraphics[height=3cm,width=.49\linewidth]{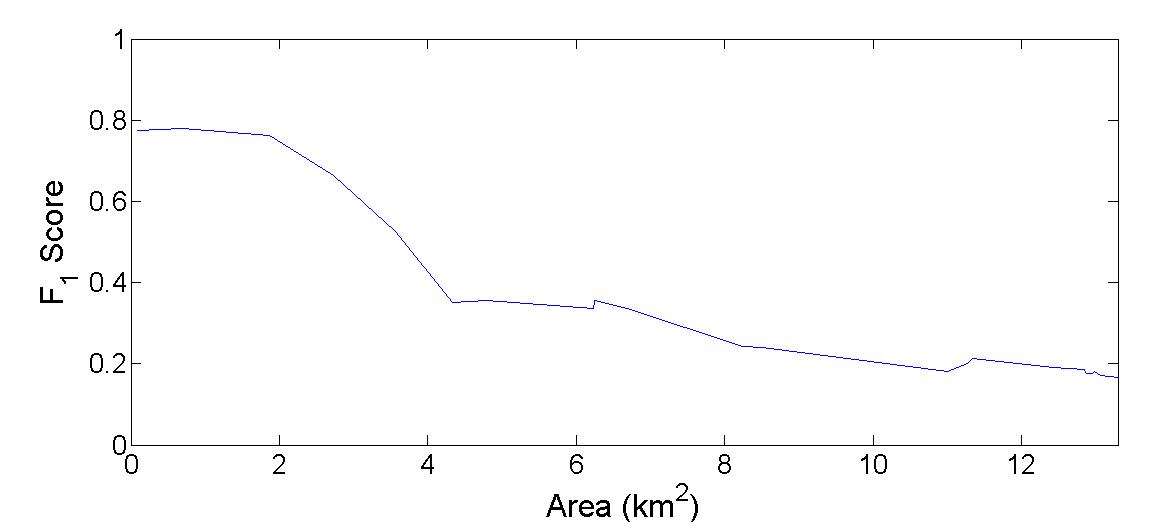}%
}\hfill
\subfloat[Edinburgh Castle\label{sfig:testa}]{%
  \includegraphics[height=3cm,width=.49\linewidth]{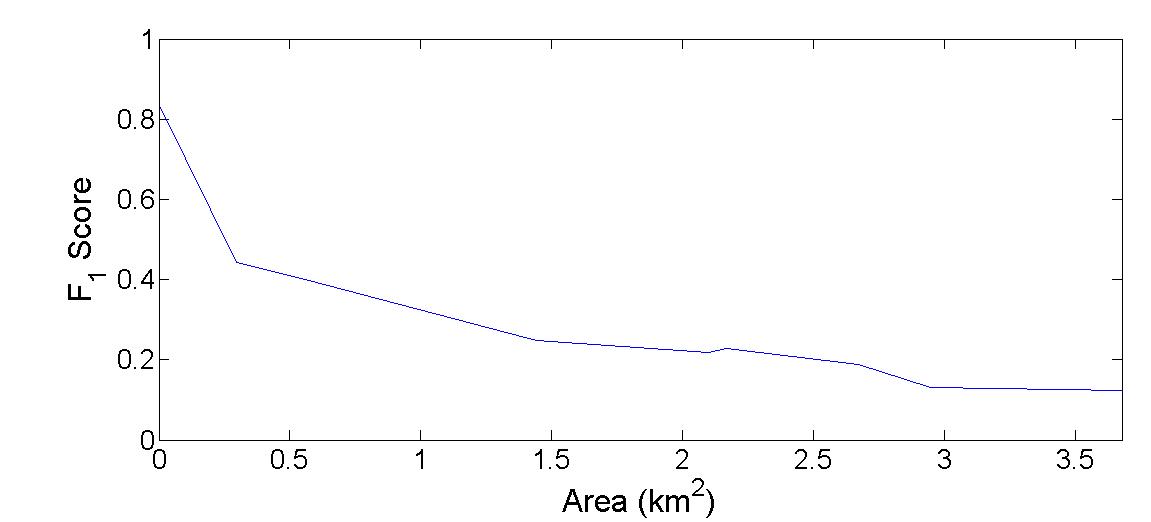}%
}\hfill
\subfloat[University of Oxford\label{sfig:testa}]{%
  \includegraphics[height=3cm,width=.49\linewidth]{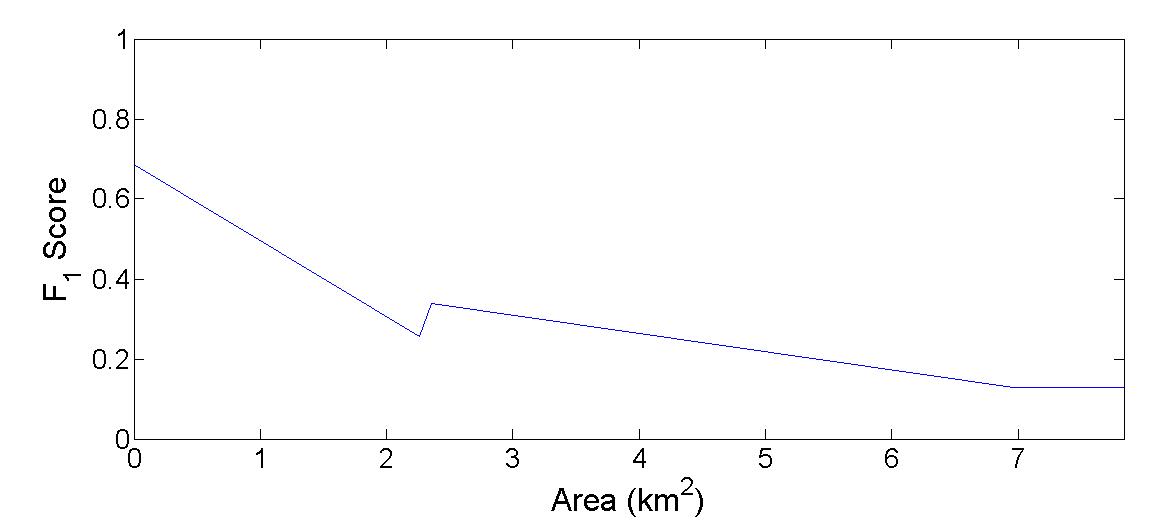}%
}
\caption{The $\mathcal{F}_1$ score according to the clusters' area}
\label{f1ToArea}
\end{figure}

In the process of discovering clusters from geo-tagged tweets relevant to a POI, the area covered by the clusters can be a matter of great interest, since several applications such as geo-marketing may desire a widespread geographic area. To illustrate this point, in Fig. \ref{f1ToArea}, we plot the $\mathcal{F}_1$ score according to the clusters' area (in km\textsuperscript{2}) for four chosen POIs. One can observe that the highest $\mathcal{F}_1$ score tends to be found when the clusters' area is very small. Therefore, although it is good to find clusters with the highest $\mathcal{F}_1$ score, it is more preferred to considerably extend the area of the resulting clusters at the expense of a slightly reduced value of $\mathcal{F}_1$ in some applications. To this end, we would like to formulate a following new performance metric expressed as the product of a power law in the clusters' area $A$ (in km\textsuperscript{2}) normalized to the area of the query region, denoted by $\bar{A} = \frac{\text{Area covered by the clusters}}{\text{Area of the query region}}$, and the $\mathcal{F}_1$ score:
\begin{align}
\label{pm}
\bar{A}^\alpha \mathcal{F}_1,
\end{align}
where $\alpha \geq 0$ is the area exponent, which balances between different levels of geographic coverage. When $\alpha$ is small, clusters with the almost highest $\mathcal{F}_1$ score are returned, and as a special case, when $\alpha=0$, our performance metric becomes the $\mathcal{F}_1$ score. On the other hand, as $\alpha$ increases, clusters covering a wide area are obtained at the cost of a reduced $\mathcal{F}_1$. Hence, given parameters for the two algorithms (i.e., ($\epsilon, N_{\text{min}}$) for DBSCAN and ($\epsilon, N_{\text{min}}, N_{\text{max}}$) for \textsf{DBSTexC}), we are able to calculate the performance metric in Equation (\ref{pm}) along with the corresponding $\mathcal{F}_1$ score and the normalized clusters' area $\bar{A}$ in each case.

\subsection{Experimental Evaluation}
We exhibit the experimental results for various values of $\alpha \geq 0$. In regard to the query region, for all chosen POIs, we assume that $\eta = 0.07$, which can also be set to other values to control the clustering quality constraint. We summarize and compare the performance of both \textsf{DBSTexC} and DBSCAN for four POIs in Table \ref{DBtex}, where $\alpha \in \{0, 0.5, 0.75,1\}$. From the table, it is evident that \textsf{DBSTexC} outperforms DBSCAN in terms of our performance metric in~(\ref{pm}) by up to 60.09\% for all four chosen POIs. The performance improvement is manifest especially for Hyde Park, which is one of the biggest and the most visited parks in London. In Figs. \ref{TexHP}--\ref{TexEd}, we show the clustering results of DBSCAN and \textsf{DBSTexC} for the four POIs when $\alpha = 0.5$. To emphasize the performance gap between the two algorithms, we illustrate the geographic cluster region with the distribution of POI-irrelevant tweets. From Fig. \ref{TexHP}, one can see that in the Hyde Park case, \textsf{DBSTexC} dramatically excludes a huge number of POI-irrelevant tweets from its clusters, while covering a much bigger geographic area in comparison with DBSCAN. This highlights the robustness of \textsf{DBSTexC} to discover high-quality clusters in terms of the proposed performance metric $\bar{A}^\alpha \mathcal{F}_1$.

\begin{table}[!t]
\captionsetup{font=normalsize}
\centering
\caption{Experimental results for DBSCAN and \textsf{DBSTexC}}

\label{DBtex}
\resizebox{\columnwidth}{!}{%
\begin{tabular}{|c|c|c|c|}
\hline
                     & \multicolumn{3}{c|}{$\bar{A}^\alpha \mathcal{F}_1$ $(\alpha = 0)$}                                                      \\ \hline
\textbf{POI name}             & DBSCAN ($X$) & \textsf{DBSTexC} ($Y$) & \begin{tabular}[c]{@{}c@{}}Improvement\\ Rate $\left(\frac{Y-X}{X} \%\right)$\end{tabular} \\ \hline
Hyde Park            & 0.7333  & 0.7391 & 0.79                                                    \\ \hline
Regent's Park        & 0.7795  & 0.7851 & 0.72                                                     \\ \hline
University of Oxford & 0.6930  & 0.6930 & 0 
\\ \hline
Edinburgh Castle     & 0.8364  & 0.8364 & 0                                                    \\ \hline
                     & \multicolumn{3}{c|}{$\bar{A}^\alpha \mathcal{F}_1$ $(\alpha = 0.5)$}                                                      \\ \hline
Hyde Park            & 0.2103  & 0.3058 & 45.41                                                    \\ \hline
Regent's Park        & 0.3184  & 0.3188 & 0.13                                                     \\ \hline
University of Oxford & 0.1288  & 0.2062 & 60.09                                                    \\ \hline
Edinburgh Castle     & 0.1333  & 0.1741 & 30.61                                                    \\ \hline
                     & \multicolumn{3}{c|}{$\bar{A}^\alpha \mathcal{F}_1$ $(\alpha = 0.75)$}                                                      \\ \hline
Hyde Park            & 0.1429  & 0.2284 & 59.83                                                    \\ \hline
Regent's Park        & 0.2216  & 0.2219 & 0.14                                                     \\ \hline
University of Oxford & 0.1288  & 0.1673 & 29.89
 \\ \hline
Edinburgh Castle     & 0.1231  & 0.1510 & 22.66
 \\ \hline
                     & \multicolumn{3}{c|}{$\bar{A}^\alpha \mathcal{F}_1$ $(\alpha = 1)$}                                                      \\ \hline
Hyde Park            & 0.1253  & 0.1816 & 44.93                                                    \\ \hline
Regent's Park        & 0.1303  & 0.1844 & 41.52
 \\ \hline
University of Oxford & 0.1288  & 0.1288 & 0
 \\ \hline
Edinburgh Castle     & 0.1231  & 0.1412 & 14.70
 \\ \hline
\end{tabular}}
\end{table}

\begin{figure}[!t]
\captionsetup[subfigure]{font=scriptsize,labelfont=normalsize}
\subfloat[DBSCAN\label{sfig:testa}]{%
  \includegraphics[height=3cm,width=.49\linewidth]{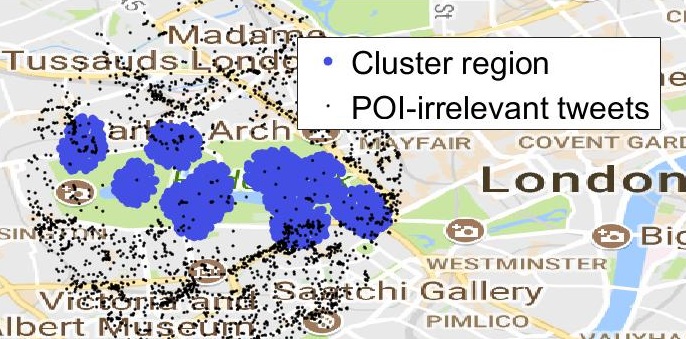}%
}\hfill
\subfloat[\textsf{DBSTexC}\label{sfig:testa}]{%
  \includegraphics[height=3cm,width=.49\linewidth]{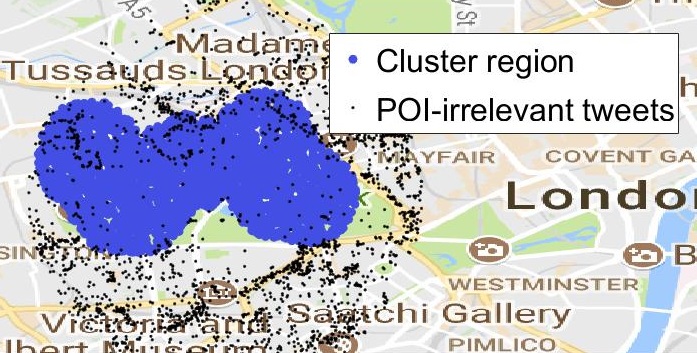}%
}
\caption{The results of DBSCAN and \textsf{DBSTexC} for Hyde Park when $\alpha =0.5$}
\label{TexHP}
\end{figure}

\begin{figure}[!t]
\captionsetup[subfigure]{font=scriptsize,labelfont=normalsize}
\subfloat[\textsf{DBSCAN}\label{sfig:testa}]{%
  \includegraphics[height=3cm,width=.49\linewidth]{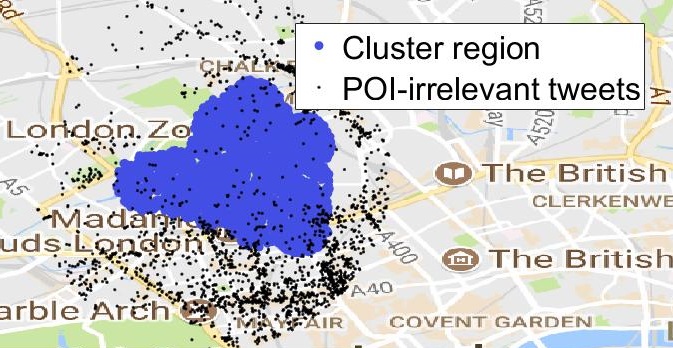}%
}\hfill
\subfloat[\textsf{DBSTexC}\label{sfig:testa}]{%
  \includegraphics[height=3cm,width=.49\linewidth]{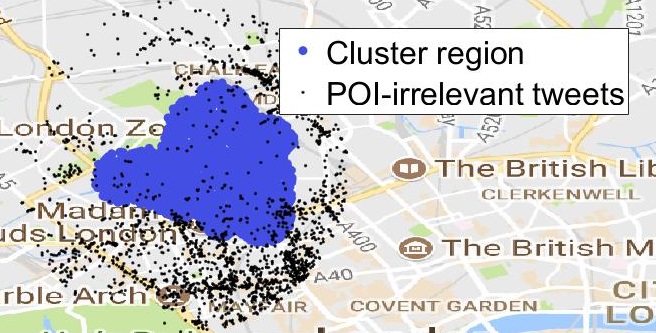}%
}
\caption{The results of DBSCAN and \textsf{DBSTexC} for Regent's Park when $\alpha =0.5$}
\label{TexRP}
\end{figure}

\begin{figure}[!t]
\captionsetup[subfigure]{font=scriptsize,labelfont=normalsize}
\subfloat[\textsf{DBSCAN}\label{sfig:testa}]{%
  \includegraphics[height=3cm,width=.49\linewidth]{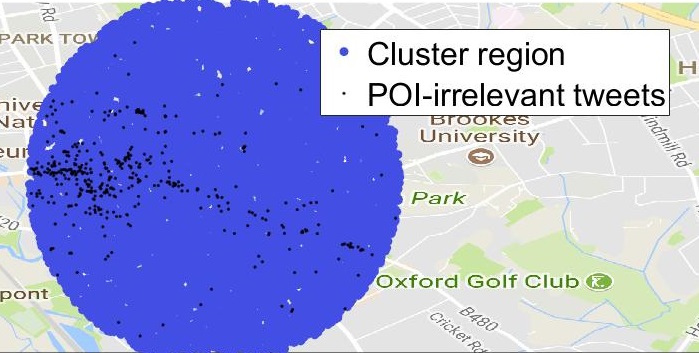}%
}\hfill
\subfloat[\textsf{DBSTexC}\label{sfig:testa}]{%
  \includegraphics[height=3cm,width=.49\linewidth]{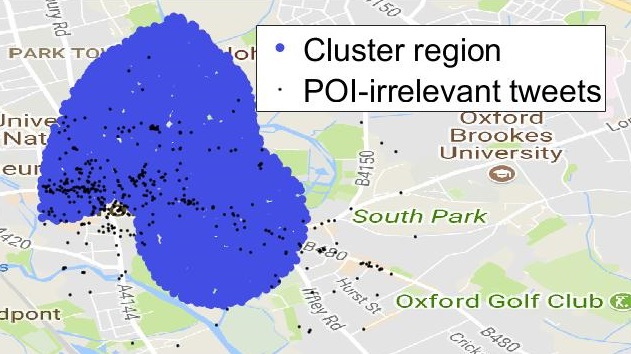}%
}
\caption{The results of DBSCAN and \textsf{DBSTexC} for University of Oxford when $\alpha =0.5$}
\label{TexOx}
\end{figure}

\begin{figure}[!t]
\captionsetup[subfigure]{font=scriptsize,labelfont=normalsize}
\subfloat[\textsf{DBSCAN}\label{sfig:testa}]{%
  \includegraphics[height=3cm,width=.49\linewidth]{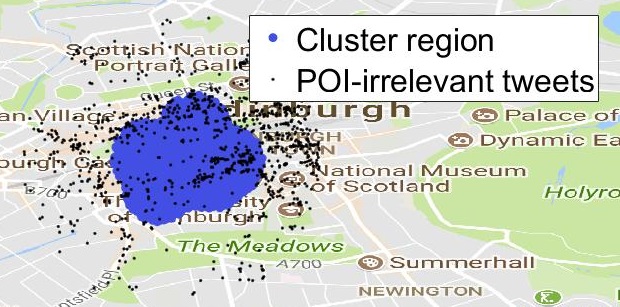}%
}\hfill
\subfloat[\textsf{DBSTexC}\label{sfig:testa}]{%
  \includegraphics[height=3cm,width=.49\linewidth]{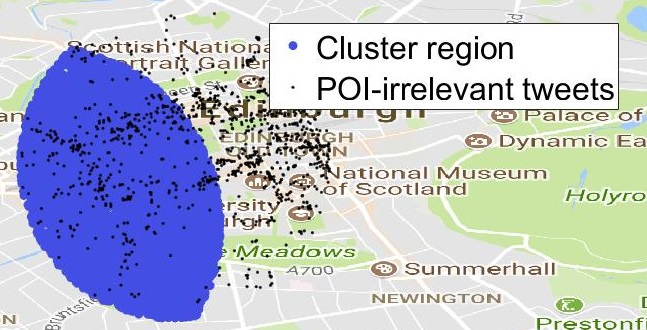}%
}
\caption{The results of DBSCAN and \textsf{DBSTexC} for Edinburgh Castle when $\alpha =0.5$}
\label{TexEd}
\end{figure}

On the other hand, for a special case where $\alpha = 0$, we notice from Table \ref{DBtex} that the \textsf{DBSTexC} algorithm has almost the same performance as that of DBSCAN. While both algorithms are able to find clusters with the high $\mathcal{F}_1$ score, it is revealed from Fig. \ref{TexAlpha0} that the clusters cover remarkably small geographic areas, which do not provide any insight or useful information about the regions where people are interested in the POIs. As a result, to obtain high-quality clusters covering large geographic areas, it is needed to incorporate the clusters' area into the performance metric. 

\begin{figure}[!t]
\captionsetup[subfigure]{font=scriptsize,labelfont=normalsize}
\subfloat[Hyde Park\label{sfig:testa}]{%
  \includegraphics[height=3cm,width=.49\linewidth]{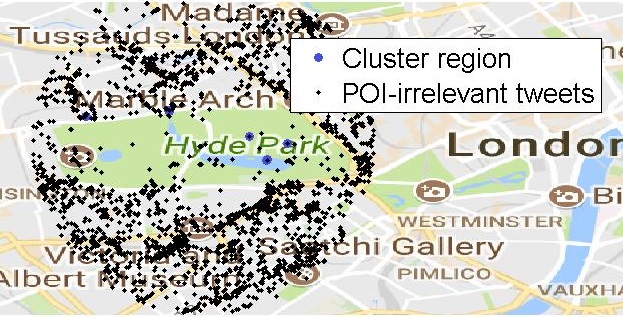}%
}\hfill
\subfloat[\textsf{Regent's Park}\label{sfig:testa}]{%
  \includegraphics[height=3cm,width=.49\linewidth]{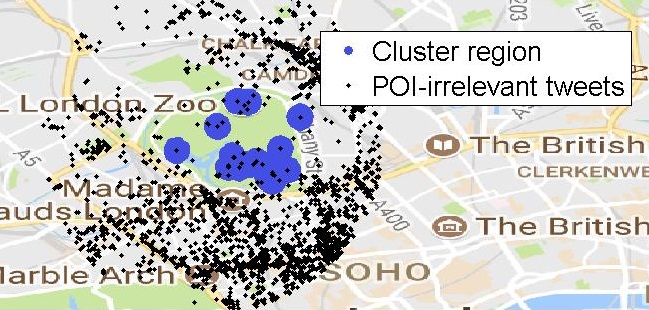}%
}\hfill
\subfloat[University of Oxford\label{sfig:testa}]{%
  \includegraphics[height=3cm,width=.49\linewidth]{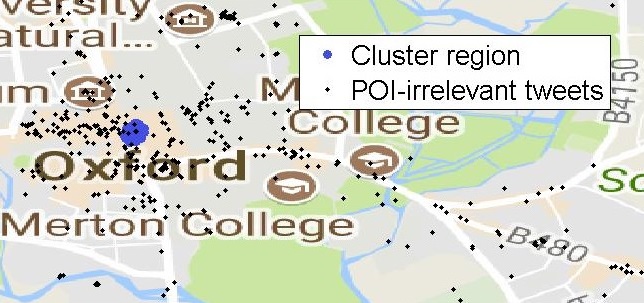}%
}\hfill
\subfloat[Edinburgh Castle\label{sfig:testa}]{%
  \includegraphics[height=3cm,width=.49\linewidth]{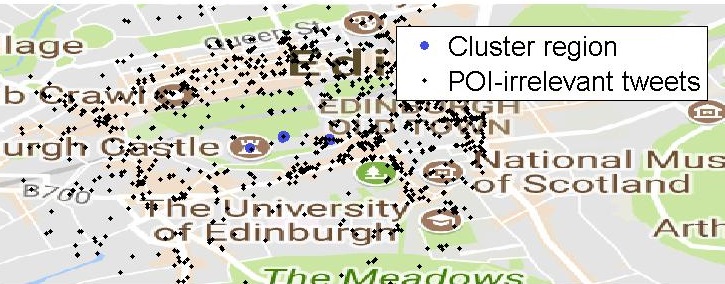}%
}
\caption{The results of \textsf{DBSTexC} when $\alpha =0$}
\label{TexAlpha0}
\end{figure}

\subsection{Computational Complexity}
We hereby analyze the computational complexity of the DBSCAN and \textsf{DBSTexC} algorithms. The runtime complexity of both algorithms is calculated by the input size (the number of tweets) times the basic operation $\epsilon$-neighborhood query (range query), which indeed dominates the complexity. 

In the case of \textsf{DBSTexC}, from Algorithms \ref{alg:Tex1} and \ref{alg:Tex2}, we can clearly see that the RangeQuery() function is invoked only for POI-relevant tweets that have not yet been visited, and the \textsf{DBSTexC} algorithm will visit every POI-relevant tweet in the dataset once. Therefore, we execute exactly one range query for every POI-relevant tweet in the dataset. For analysis, let $Q$ denote the complexity of the function range query, and $n$ and $m$ denote the number of POI-relevant and irrelevant tweets, respectively. It then follows that the complexity is expressed as  $\mathcal{O}(n \cdot Q)$. Based on how the function RangeQuery() is implemented, its complexity analysis can be divided into the following two cases:
\begin{itemize}
\item If the range query is implemented using a {\em linear scan}, then we have $Q = \mathcal{O}( (n+m) \cdot D)$, where $D$ indicates the cost of computing the distance between two points. Because each geo-tagged tweet in our dataset has a two-dimensional coordinate and is represented by a 64-bit data type in the database, the cost $D$ can be treated as a constant, independent of $n$ and $m$. Hence, the complexity of the range query and \textsf{DBSTexC} are $\mathcal{O}(n+m)$ and $\mathcal{O}(n^2+nm)$, respectively.
\item If the range query is implemented using a {\em spatial index}, then we can calculate the worst-case runtime complexity by analyzing both the cost of building the index and the worst-case complexity of the function RangeQuery() used along with the spatial index. For example, for a two-dimensional tree, the worst-case complexity of RangeQuery() is $\mathcal{O}(n+m)$, and the cost of building a two-dimensional tree from $n+m$ geo-tagged points is  
\begin{flalign*}
&\mathcal{O}((n+m)\cdot \log(n+m)) \\ & = \mathcal{O}((n+m)\cdot [\log n + \log (1+ \frac{m}{n})]) 
\\&= \mathcal{O}((n+m)\cdot \log n),
\end{flalign*}
where the last equality holds under the assumption that $m=n^\beta$ for $\beta\ge1$. Therefore, it follows that the time complexity of \textsf{DBSTexC} is $\mathcal{O}(n\cdot (n+m) + (n+m) \cdot \log n) = \mathcal{O}(n^2 + nm)$.
\end{itemize}

For the DBSCAN algorithm, it has recently been proved in \cite{schubert2017dbscan} that the worst-case complexity is  $\mathcal{O}(n \cdot Q)$. Based on the arguments above, when the range query is implemented using a linear scan, the complexity is $\mathcal{O}( n^2 \cdot D) = \mathcal{O} (n^2)$. On the contrary, if the range query is accelerated using a spatial index such as a two-dimensional tree, the worst-case runtime complexity of DBSCAN is  $\mathcal{O}(n^2)$ since it takes $O(n\log n)$ to build the tree from $n$ geo-tagged points and the range query has the worst-case complexity of $O(n)$.

To summarize the aforementioned analysis, the worst-case time complexity of \textsf{DBSTexC} and DBSCAN is $\mathcal{O}(n^2 + nm)$ and $\mathcal{O}(n^2)$, respectively. If we focus on a region where $m = c \cdot n$ for a constant $c>0$, then the complexity of \textsf{DBSTexC} is  $\mathcal{O}(n^2)$. In the other region where $m = n^\beta$ for $\beta > 1$, the the complexity of \textsf{DBSTexC} is $\mathcal{O}(n^{1+\beta})$. 

\begin{figure}[!t]
\captionsetup[subfigure]{font=scriptsize,labelfont=normalsize}
\subfloat[Hyde Park\label{sfig:testa}]{%
  \includegraphics[height=3cm,width=.49\linewidth]{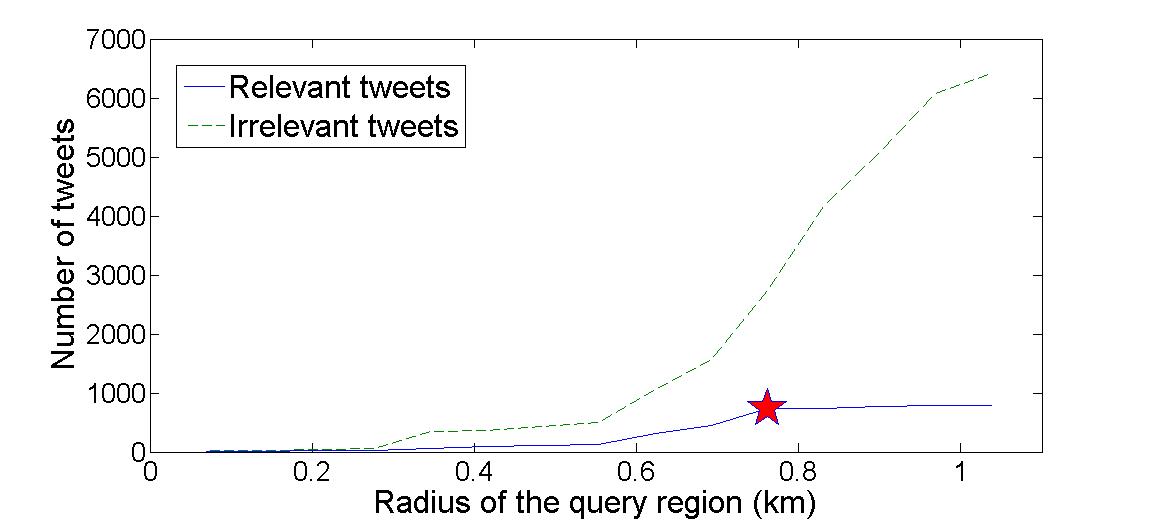}%
}\hfill
\subfloat[Regent's Park\label{sfig:testa}]{%
  \includegraphics[height=3cm,width=.49\linewidth]{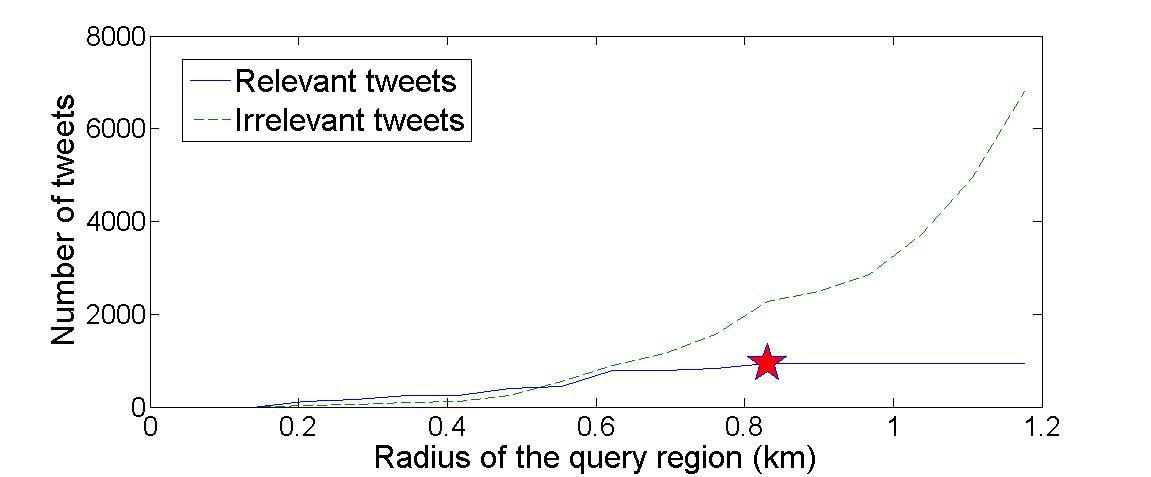}%
}\hfill
\subfloat[Edinburgh Castle\label{sfig:testa}]{%
  \includegraphics[height=3cm,width=.49\linewidth]{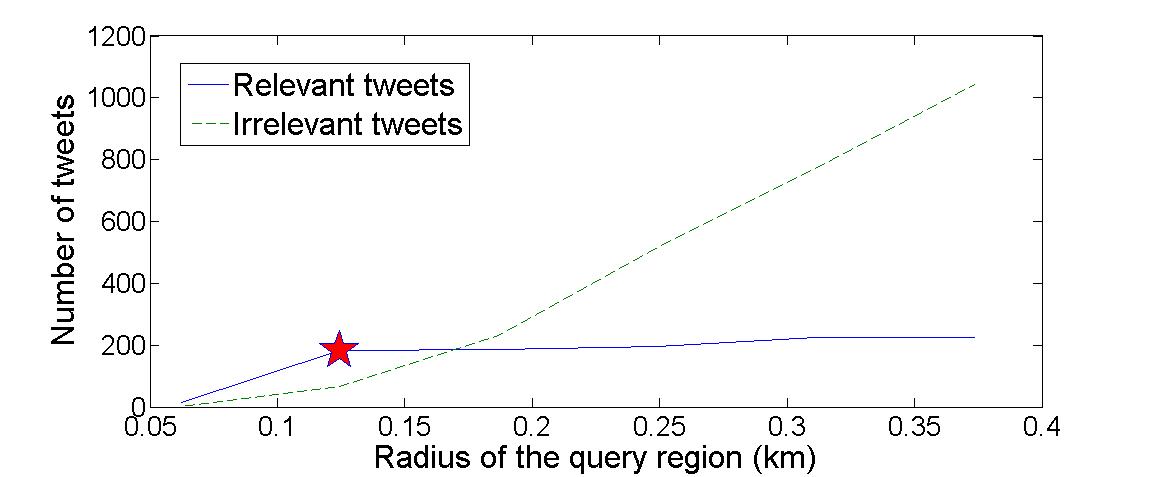}%
}\hfill
\subfloat[University of Oxford\label{sfig:testa}]{%
  \includegraphics[height=3cm,width=.49\linewidth]{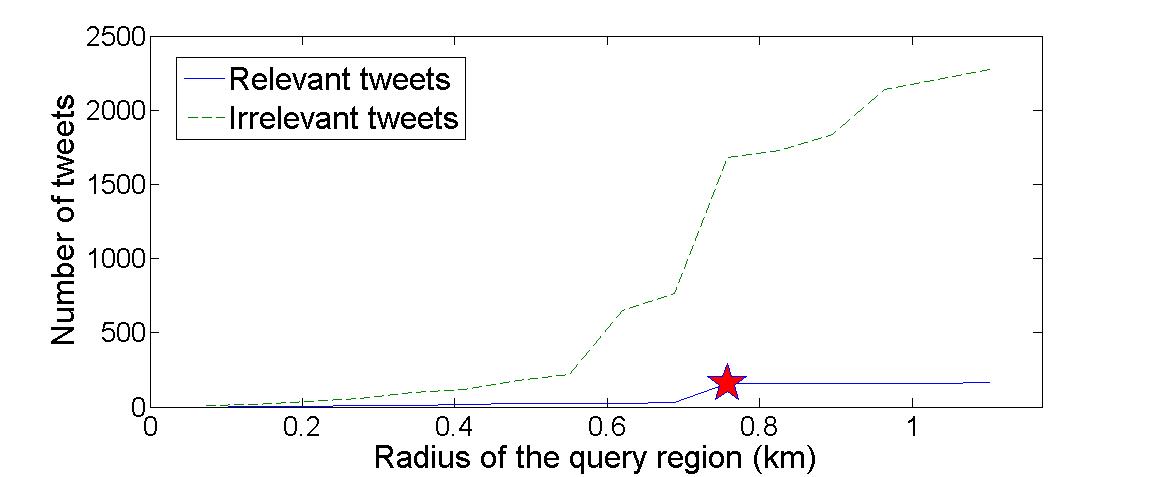}%
}
\caption{The number of tweets according to the radius of the query region}
\label{NoTweetToRadius}
\end{figure}

To numerically validate our complexity analysis, we first plot the number of tweets according to different radii of the query region. From Fig. \ref{NoTweetToRadius}, we observe a common trend that the numbers of POI-relevant and POI-irrelevant tweets, denoted by $n$ and $m$, respectively, increase with the increasing radius of the query region. However, their rates of growth are different; up to a certain radius of the query region, the numbers of POI-relevant and the POI-irrelevant tweets grow at a similar rate, but beyond such a radius (depicted in the figure with a star), the number of POI-irrelevant tweets grows faster than the number of POI-relevant tweets. This observation is basically consistent with our prior assumption: there is a region where the number of POI-irrelevant tweets is a constant times the number of POI-relevant tweets, having the complexity of $\mathcal{O}(n^2)$ for \textsf{DBSTexC}; and there is another region where the rate of growth of the number of POI-irrelevant tweets is higher than that of the POI-relevant tweets, having the complexity of $\mathcal{O}(n^{1+\beta})$ for $\beta > 1$ for \textsf{DBSTexC}. 

\begin{figure}[!t]
\captionsetup[subfigure]{font=scriptsize,labelfont=normalsize}
\subfloat[Hyde Park\label{sfig:testa}]{%
  \includegraphics[height=3cm,width=.49\linewidth]{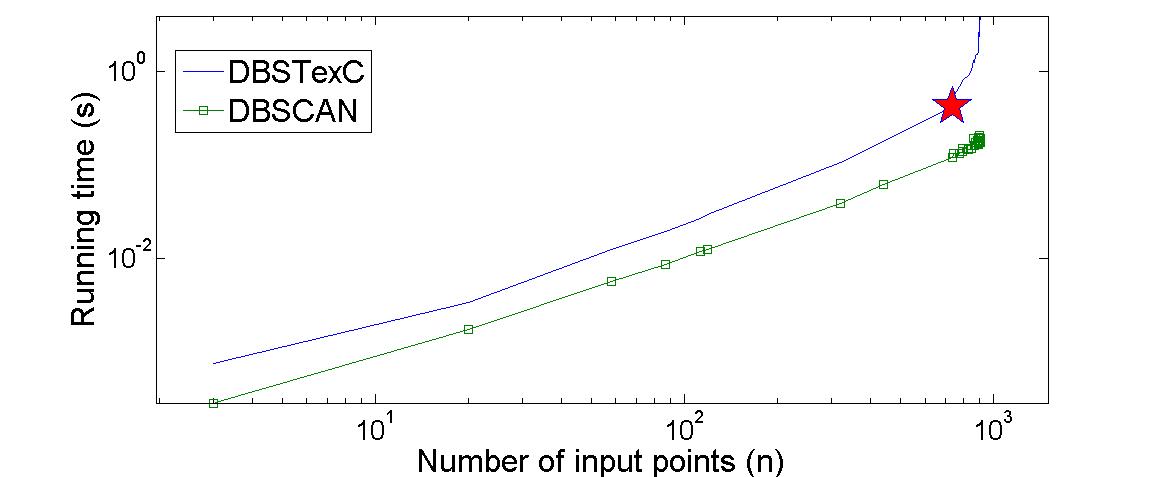}%
}\hfill
\subfloat[Regent's Park\label{sfig:testa}]{%
  \includegraphics[height=3cm,width=.49\linewidth]{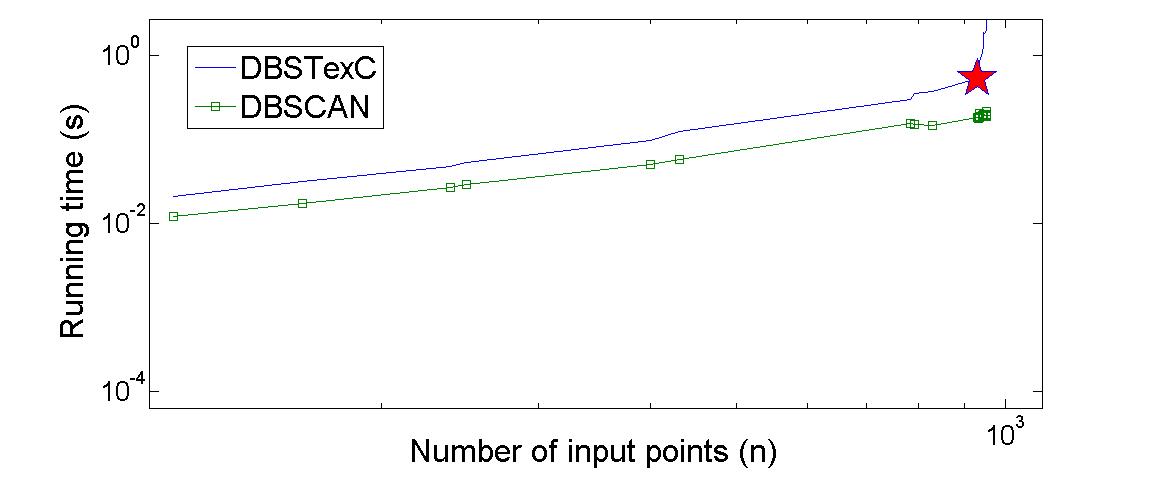}%
}\hfill
\subfloat[Edinburgh Castle\label{sfig:testa}]{%
  \includegraphics[height=3cm,width=.49\linewidth]{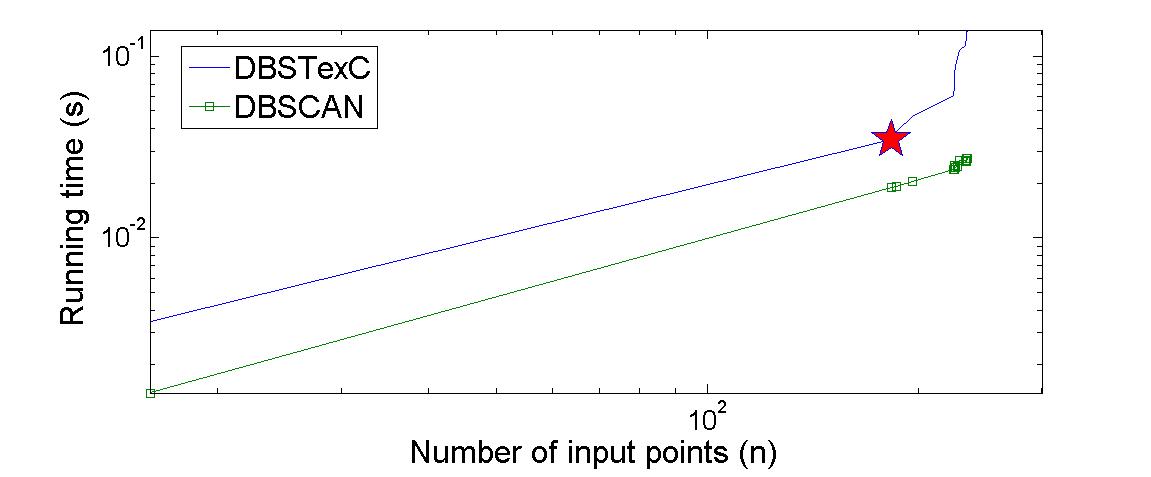}%
}\hfill
\subfloat[University of Oxford\label{sfig:testa}]{%
  \includegraphics[height=3cm,width=.49\linewidth]{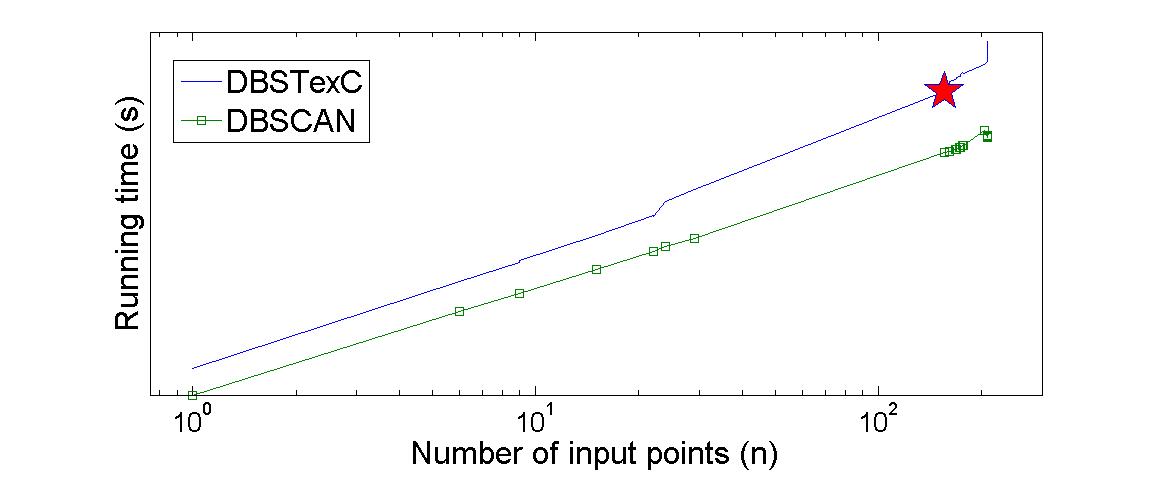}%
}
\caption{The runtime complexity of \textsf{DBSTexC} and DBSCAN}
\label{TexCANRuntime}
\end{figure}

We further validate our complexity analysis by plotting the actual runtime complexity of the \textsf{DBSTexC} and DBSCAN algorithm for the worst case. It is easily seen that the worst case takes place when the parameters of \textsf{DBSTexC} and DBSCAN are set to extreme values corresponding to ($\epsilon, N_{\text{min}}$) = (radius of the query region, 1) for DBSCAN and ($\epsilon, N_{\text{min}}, N_{\text{max}}$) = (radius of the query region, 1, total number of POI-irrelevant tweets) for \textsf{DBSTexC}. Under this parameter setting, Fig. \ref{TexCANRuntime} numerically shows the runtime complexity of the \textsf{DBSTexC} and DBSCAN algorithms in log-log scale according to four different POIs. From Fig. \ref{TexCANRuntime}, we clearly see that up to a certain value of the number of geo-tagged tweets, the \textsf{DBSTexC} and DBSCAN have a similar rate of growth maintaining a constant gap between each other. Beyond the point (depicted in the figure with a star), the time complexity of \textsf{DBSTexC} is higher than that of DBSCAN. Compared with Fig. \ref{NoTweetToRadius}, these transitional points exactly match the ones dividing our query region into two sub-regions corresponding to $m = c \cdot n$ for a constant $c$ and $m = n^\beta$ for $\beta > 1$. Therefore, from Figs. \ref{NoTweetToRadius} and \ref{TexCANRuntime}, it is possible to adequately substantiate our analysis on the complexity of the \textsf{DBSTexC} and DBSCAN algorithms.

\section{Fuzzy DBSTexC (\textsf{F-DBSTexC})}
\label{fuzzy}
Thus far, the \textsf{DBSTexC} algorithm has been designed by finding clusters with strict boundaries. For further analysis, we study the geographic distribution of tweets (i.e., two-dimensional coordinates) by using the sorted $k$-th-nearest neighbor ($k$-NN) distance plot, which shows the distance from geo-tagged points to their $k$-th-nearest neighbors sorted in ascending order. If there exists a sudden and sharp increase in the distances between geo-tagged points, then it indicates that clusters and noise points are clearly separated. On the other hand, if we observe a smooth increase in the distances between tweets, then it may not be clear which tweets should be grouped as clusters and which tweets should be treated as noise. In other words, decision boundaries for clusters would be fuzzy. In Fig. \ref{fig:fig4}, the $k$-NN distance plot for the four POIs is shown when $k=4$. From the figure, we observe that the geographic distribution of tweets is generally smooth. For this reason, using crisp boundaries to separate clusters may not exploit the entire geographic features of the data. To overcome this problem, we hereby propose an extension of \textsf{DBSTexC}, called Fuzzy \textsf{DBSTexC} (\textsf{F-DBSTexC}), which incorporates the notion of fuzzy clustering into \textsf{DBSTexC} with a view to fully capturing the smoothly distributed geographic characteristics of tweets. 

\begin{figure}[!t]
\centering
\captionsetup[subfigure]{font=scriptsize,labelfont=normalsize}
\subfloat[Hyde Park\label{sfig:testa}]{%
  \includegraphics[height=4cm,width=.49\linewidth]{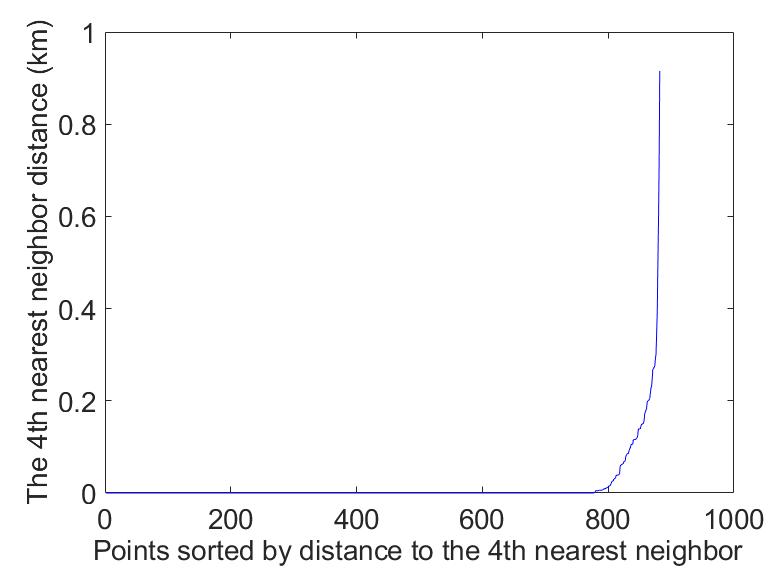}%
}\hfill
\subfloat[Regent's Park\label{sfig:testa}]{%
  \includegraphics[height=4cm,width=.49\linewidth]{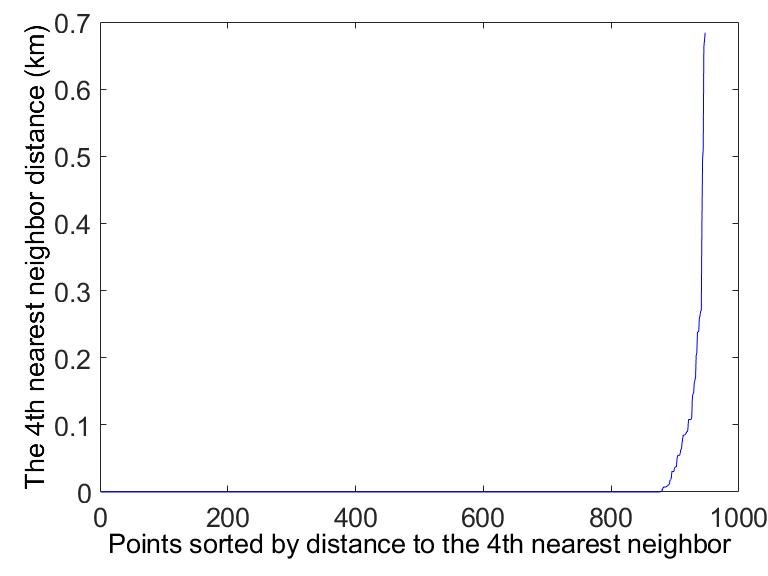}%
}\hfill
\subfloat[Edinburgh Castle\label{sfig:testa}]{%
  \includegraphics[height=4cm,width=.49\linewidth]{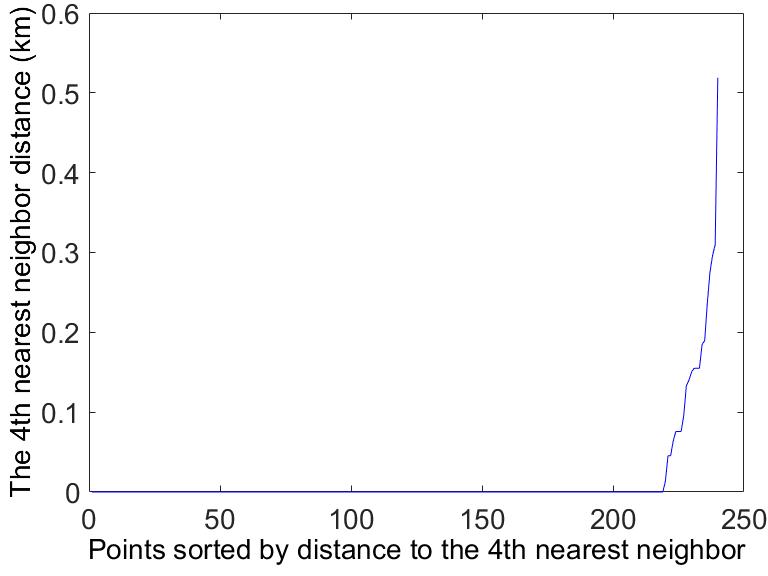}%
}\hfill
\subfloat[University of Oxford\label{sfig:testa}]{%
  \includegraphics[height=4cm,width=.49\linewidth]{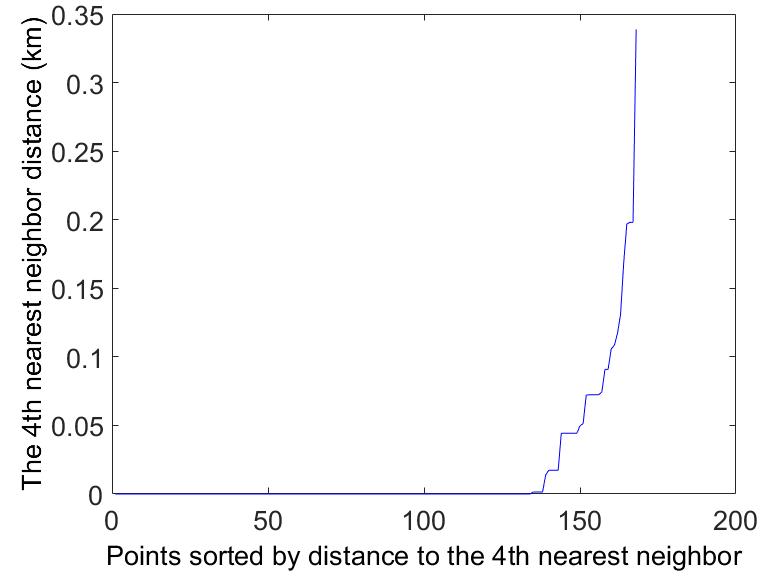}%
}
\caption{The $k$-NN distance plot for different POIs when $k=4$}
\label{fig:fig4}
\end{figure}

\subsection{\textsf{F-DBSTexC} Algorithm}
To design a new algorithm with the notion of fuzzy clustering, we relax the constraints on a point's neighborhood density. That is, we replace the parameters $N_{\text{min}}$ and $N_{\text{max}}$ by two new sets of parameters ($N_{\text{min}_1}$, $N_{\text{min}_2}$) and  ($N_{\text{max}_1}$, $N_{\text{max}_2}$), respectively, which specify the soft constraints on a point's neighborhood density. For example, in an $\epsilon$-neighborhood of a POI-relevant tweet, if the number of POI-relevant tweets is larger than $N_{\text{min}_1}$ and the number of POI-irrelevant tweets is smaller than $N_{\text{max}_2}$, then a fuzzy neighborhood is generated. To determine the neighborhood cardinality, we introduce monotonically non-decreasing membership functions $J_{Re}(p)$ and $J_{Irre}(p)$ for the POI-relevant tweets and POI-irrelevant tweets, respectively, as follows~\cite{ienco2016fuzzy}:\footnote{Other types of membership functions~\cite{ienco2016fuzzy} can also be applicable.}

\begin{align}
\label{jre}
J_{Re}(p) = 
\begin{cases}
1 &\mbox{if } |X_\epsilon (p)| \ge N_{\text{min}_2} \\
\frac{|X_\epsilon (p)| - N_{\text{min}_1}}{N_{\text{min}_2} - N_{\text{min}_1}} &\mbox{if } N_{\text{min}_1} \le |X_\epsilon (p)| \le N_{\text{min}_2} \\
0 &\mbox{if } |X_\epsilon (p)| \le N_{\text{min}_1}, \\
\end{cases}
\end{align}

\begin{align}
\label{jirre}
J_{Irre}(p) = 
\begin{cases}
1 &\mbox{if } |Y_\epsilon (p)| \le N_{\text{max}_1} \\
\frac{N_{\text{max}_2} - |Y_\epsilon (p)|}{N_{\text{max}_2} - N_{\text{max}_1}} &\mbox{if } N_{\text{max}_1} \le |Y_\epsilon (p)| \le N_{\text{max}_2} \\
0 &\mbox{if } |Y_\epsilon (p)| \ge N_{\text{max}_2}, \\
\end{cases}
\end{align}
where $|X_\epsilon (p)|$ and $|Y_\epsilon (p)|$ denote the number of  POI-relevant and POI-irrelevant tweets, respectively, in a neighborhood of point $p$. The final cardinality of the $\epsilon$-neighborhood of a point $p$ is then given by 
\begin{align}
\label{mup}
\mu_p = \frac{1}{2}[J_{Re}(p) + J_{Irre}(p)].
\end{align}

Based on this notation, the definition of a core point in Definition 3 is revised as below.

\textit{Definition 9 (Core point):} A point $p\in \mathcal{X}$ is a core point if it fulfills the following condition:
\begin{equation*}
|\mathcal{X}_{\epsilon}(p)| \ge N_{\text{min}_1} \; \text{and} \; |\mathcal{Y}_{\epsilon}(p)| \le N_{\text{max}_2}.
\end{equation*}

Next, the \textsf{F-DBSTexC} algorithm is specified in Algorithms~\ref{alg:FTex1} and~\ref{alg:FTex2}. Compared to the original \textsf{DBSTexC}, modified parts correspond to line~\ref{change1} of Algorithm~\ref{alg:FTex1} and line 6 of Algorithm~\ref{alg:FTex2}, which serve to relax the constraints on a point's neighborhood density. The \textsf{F-DBSTexC} algorithm adds points to the clusters with their distinct fuzzy score $\mu_{p}$, as expressed in line 9 of Algorithm~\ref{alg:FTex2}.

\begin{algorithm}[!t]
\small
\caption{\textsf{F-DBSTexC}($\mathcal{X}$,$\mathcal{Y}$, $\epsilon$, $N_{\text{min}_1}$, $N_{\text{min}_2}$, $N_{\text{max}_1}$, $N_{\text{max}_2}$)}
\label{alg:FTex1}
\begin{algorithmic}
\Require $\mathcal{X}$,$\mathcal{Y}$, $\epsilon$, $N_{\text{min}_1}$, $N_{\text{min}_2}$, $N_{\text{max}_1}$, $N_{\text{max}_2}$
\Ensure Clusters with different labels $C$
\algrenewcommand\algorithmicensure{\textbf{Initialization:}}
\Ensure $C\gets 0$; $n\gets |\mathcal{X}|$; $m\gets |\mathcal{Y}|$; $p_i$ is a point in the set $\mathcal{X}$
\For{each $p_i$}
\If{$p_i$ is not visited}
\State Mark $p_i$ as visited
\State $[\mathcal{X}_\epsilon (p_i), \mathcal{Y}_\epsilon (p_i)] =$ RangeQuery$(p_i)$
\If{$|\mathcal{X}_\epsilon (p_i)|\ge N_{\text{min}_1} \  \& \  |\mathcal{Y}_\epsilon (p_i)|\le N_{\text{max}_2}$}  \label{change1}
\State $C\gets C+1$
\State ExpandCluster$(p_i,\mathcal{X}_\epsilon (p_i),\mathcal{Y}_\epsilon (p_i))$
\EndIf
\EndIf
\EndFor 
\end{algorithmic}
\end{algorithm}

\begin{algorithm}[!t]
\small
\begin{algorithmic}
\caption{ExpandCluster$(p_i, \mathcal{X}_\epsilon (p_i), \mathcal{Y}_\epsilon (p_i))$}
\label{alg:FTex2}
\Require $p_i, \mathcal{X}_\epsilon (p_i), \mathcal{Y}_\epsilon (p_i)$
\Ensure Cluster $C$ with all of its members
\algrenewcommand\algorithmicensure{\textbf{Initialization:}}
\State Add $p_i$ to the current cluster with fuzzy score $\mu_{p_i}$
\For{each point $p_j$ in the set $\mathcal{X}_\epsilon (p_i)$}
\If {$p_j$ is not visited}
\State Mark $p_j$ as visited
\State $[\mathcal{X}_\epsilon (p_j), \mathcal{Y}_\epsilon (p_j)] =$ RangeQuery$(p_j)$
\If{$|\mathcal{X}_\epsilon (p_j)|\ge N_{\text{min}_1}\  \& \  |\mathcal{Y}_\epsilon (p_j)|\le N_{\text{max}_2}$}  
\label{FTex:change2}
\State $\mathcal{X}_\epsilon (p_i) = \mathcal{X}_\epsilon (p_i) \cup \mathcal{X}_\epsilon (p_j)$
\State $\mathcal{Y}_\epsilon (p_i) = \mathcal{Y}_\epsilon (p_i) \cup \mathcal{Y}_\epsilon (p_j)$.  
\State Add $p_j$ to the current cluster with fuzzy score $\mu_{p_j}$   \label{FTex:change3}
\EndIf
\EndIf
\If{$p_j$ does not have a label}
\State Add $p_j$ to the current cluster
\EndIf
\EndFor
\If{$|\mathcal{Y}_\epsilon (p_i)|\ne 0$}
\For{each point $q_j$ in the set $\mathcal{Y}_\epsilon (p_i)$}
\If{$q_j$ is not visited}
\State Mark $q_j$ as visited
\If{$q_j$ does not have a label}
\State Add $q_j$ to the current cluster
\EndIf
\EndIf
\EndFor
\EndIf 
\end{algorithmic}
\end{algorithm}

\subsection{Experimental Evaluation}

\begin{table}[!t]
\captionsetup{font=normalsize}
\centering
\caption{Experimental results for \textsf{DBSTexC} and \textsf{F-DBSTexC}}

\label{FDBTexC}
\resizebox{\columnwidth}{!}{%
\begin{tabular}{|c|c|c|c|}
\hline
                     & \multicolumn{3}{c|}{$\bar{A}^\alpha \mathcal{F}_1$ $(\alpha = 0)$}                                                      \\ \hline
\textbf{POI name}             & \textsf{DBSTexC} ($X$)& \textsf{F-DBSTexC} ($Y$) & \begin{tabular}[c]{@{}c@{}}Improvement\\ Rate $\left(\frac{Y-X}{X} \%\right)$\end{tabular} \\ \hline
Hyde Park            & 0.7391  & 0.7556 & 2.23
 \\ \hline
Regent's Park        & 0.7851  & 0.7949 & 1.25
 \\ \hline
University of Oxford & 0.6930  & 0.7186 & 3.69                                                    \\ \hline
Edinburgh Castle     & 0.8364  & 0.8503 & 1.66                                                    \\ \hline
                     & \multicolumn{3}{c|}{$\bar{A}^\alpha \mathcal{F}_1$ $(\alpha = 0.5)$}                                                      \\ \hline
Hyde Park            & 0.3058  & 0.3063 & 0.16                                                    \\ \hline
Regent's Park        & 0.3188  & 0.3325 & 4.30                                                     \\ \hline
University of Oxford & 0.2062  & 0.2403 & 16.54                                                    \\ \hline
Edinburgh Castle     & 0.1741  & 0.1874 & 7.64                                                    \\ \hline
                     & \multicolumn{3}{c|}{$\bar{A}^\alpha \mathcal{F}_1$ $(\alpha = 0.75)$}                                                      \\ \hline
Hyde Park            & 0.2284  & 0.2302 & 0.79                                                    \\ \hline
Regent's Park        & 0.2219  & 0.2228 & 0.41                                                     \\ \hline
University of Oxford & 0.1673  & 0.1808 & 8.07                                                     \\ \hline
Edinburgh Castle     & 0.1510  & 0.1662 & 10.01                                                     \\ \hline
                     & \multicolumn{3}{c|}{$\bar{A}^\alpha \mathcal{F}_1$ $(\alpha = 1)$}                                                      \\ \hline
Hyde Park            & 0.1816  & 0.1896 & 4.41                                                    \\ \hline
Regent's Park        & 0.1844  & 0.1848 & 0.22
 \\ \hline
University of Oxford & 0.1288  & 0.1640 & 27.33
 \\ \hline
Edinburgh Castle     & 0.1412  & 0.1412 & 0                                                     \\ \hline
\end{tabular}}
\end{table}

We summarize the experimental results in Table \ref{FDBTexC} according to different values of $\alpha \geq 0$. From the table, one can make the following insightful observations:
\begin{itemize}
\item The clustering quality of \textsf{F-DBSTexC} is greater than or at least equal to that of \textsf{DBSTexC} for all chosen POIs, showing the performance gain over \textsf{DBSTexC} by up to 27.33\%.
\item Although \textsf{F-DBSTexC} has slightly better performance than that of \textsf{DBSTexC} for the two POIs located in London (i.e., Hyde Park and Regent's Park), it remarkably outperforms \textsf{DBSTexC} for POIs in smaller cities such as University of Oxford and Edinburgh Castle.
\end{itemize}
The first observation can be easily understood because \textsf{F-DBSTexC} is a fuzzy extension of \textsf{DBSTexC}; therefore its performance is guaranteed to be at least as good as that of \textsf{DBSTexC}. On the other hand, the second observation may not be straightforward. We scrutinize the geographic distribution of tweets in various locations and notice that in general, POIs in crowded cities like London are surrounded by a significant number of POI-irrelevant tweets. As a result, further extension of the clusters' area would not be beneficial. However, for POIs in smaller cities such as Oxford and Edinburgh, the geographic distribution of POI-irrelevant tweets around a POI tends to be much more sparse, enabling fuzzy extension of \textsf{DBSTexC} to work effectively. To verify our observation, we conduct additional experiments for four different POIs both in populous metropolitan areas and smaller cities. The experimental results are summarized in Table \ref{FDBTexC_cont}. Among the four newly chosen POIs, Buckingham Palace and Greenwich Park are located in London; Cambridge University and Glasgow University are in the city of Cambridge and Glasgow, respectively. One can see that for POIs in London, \textsf{F-DBSTexC} shows a slightly better clustering quality than that of \textsf{DBSTexC}. However, for POIs in Cambridge and Glasgow, two smaller cities, \textsf{F-DBSTexC} is much superior to \textsf{DBSTexC}. This remark highlights our proposition that \textsf{F-DBSTexC} is a dynamic extension of \textsf{DBSTexC}, allowing \textsf{DBSTexC} to apply in different situations with diverse types  of POIs.  

\begin{table}[!t]
\captionsetup{font=normalsize}
\centering
\caption{Additional experimental results for \textsf{DBSTexC} and \textsf{F-DBSTexC}}

\label{FDBTexC_cont}
\resizebox{\columnwidth}{!}{%
\begin{tabular}{|c|c|c|c|}
\hline
                     & \multicolumn{3}{c|}{$\bar{A}^\alpha \mathcal{F}_1$ $(\alpha = 0)$}                                                      \\ \hline
\textbf{POI name}             & \textsf{DBSTexC} ($X$) & \textsf{F-DBSTexC} ($Y$) & \begin{tabular}[c]{@{}c@{}}Improvement\\ Rate  $\left(\frac{Y-X}{X} \%\right)$\end{tabular} \\ \hline
Buckingham Palace    & 0.7594  & 0.7611 & 0.22                                                    \\ \hline
Greenwich Park       & 0.7445  & 0.7588 & 1.92                                                     \\ \hline
Cambridge University & 0.6495  & 0.6495 & 0                                                    \\ \hline
Glasgow University   & 0.6839  & 0.7000 & 2.35                                                    \\ \hline
                     & \multicolumn{3}{c|}{$\bar{A}^\alpha \mathcal{F}_1$ $(\alpha = 0.5)$}                                                      \\ \hline
Buckingham Palace    & 0.2651  & 0.2652 & 0.04                                                    \\ \hline
Greenwich Park       & 0.2080  & 0.2081 & 0.05                                                     \\ \hline
Cambridge University & 0.0951  & 0.1084 & 13.98                                                    \\ \hline
Glasgow University   & 0.1443  & 0.1770 & 22.66                                                    \\ \hline
                     & \multicolumn{3}{c|}{$\bar{A}^\alpha \mathcal{F}_1$ $(\alpha = 0.75)$}                                                      \\ \hline
Buckingham Palace    & 0.2011  & 0.2024 & 0.65                                                    \\ \hline
Greenwich Park       & 0.1742  & 0.1743 & 0.06                                                     \\ \hline
Cambridge University & 0.0830  & 0.0866 & 4.34                                                    \\ \hline
Glasgow University   & 0.0749  & 0.0771 & 2.94                                                    \\ \hline
                     & \multicolumn{3}{c|}{$\bar{A}^\alpha \mathcal{F}_1$ $(\alpha = 1)$}                                                      \\ \hline
Buckingham Palace    & 0.1590  & 0.1595 & 0.31                                                  \\ \hline
Greenwich Park       & 0.1586  & 0.1588 & 0.13                                                   \\ \hline
Cambridge University & 0.0846  & 0.0863 & 2.01                                                     \\ \hline
Glasgow University   & 0.0722  & 0.0754 & 4.43                                                    \\ \hline
\end{tabular}}
\end{table}

\subsection{Computational Complexity}
Compared to \textsf{DBSTexC}, \textsf{F-DBSTexC} relaxes the constraints on  a point's neighborhood density. However, the computational complexity of \textsf{F-DBSTexC} is still dominated by the function RangeQuery(), and \textsf{F-DBSTexC} invokes the function exactly once for every POI-relevant data point. Therefore, the computational complexity of \textsf{F-DBSTexC} is of the same order as that of \textsf{DBSTexC}, which is  $\mathcal{O}(n^2 + nm)$. More specifically, the complexity of \textsf{F-DBSTexC} is  $\mathcal{O}(n^2)$ in a region where $m=c\times n$ for a constant $c$, and it follows $\mathcal{O}(n^{1+\beta})$  in another region where $m=n^\beta$ for $\beta>1$.

\section{Concluding Remarks}\label{SEC:Conclusion}
\label{summary}
As a generalized version of DBSCAN, we introduced \textsf{DBSTexC}, a new spatial clustering algorithm that further leverages textual information on Twitter, composed of $n$ POI-relevant tweets and $m$ POI-irrelevant tweets. The algorithm is beneficial when we aim to find clusters from geo-tagged tweets which are heterogeneous in terms of textual description since \textsf{DBSTexC} effectively excludes regions containing a huge number of undesired POI-irrelevant tweets. The computational complexity of \textsf{DBSTexC} was shown to be $\mathcal{O}(n^2)$ in a region where $m = c \cdot n$ for a constant $c>0$, and $\mathcal{O}(n^{1+\beta})$ in the other region where $m = n^\beta$ for $\beta > 1$. We demonstrated the performance of \textsf{DBSTexC} to be far superior to that of DBSCAN in terms of our performance metric $\bar{A}^\alpha \mathcal{F}_1$, where $\alpha\geq0$ is the area exponent. As a further extension, we introduced \textsf{F-DBSTexC}, which incorporates the notion of fuzzy clustering into \textsf{DBSTexC}. By fully capturing their geographic features, the \textsf{F-DBSTexC} algorithm was shown to outperform the original \textsf{DBSTexC} for the POIs located in sparsely-populated cities. The design methodology that \textsf{DBSTexC} and \textsf{F-DBSTexC} provide takes an important step towards a better understanding of jointly utilizing spatial and textual information in designing density-based clustering and towards a broad range of applications from geo-marketing to location-based services such as geo-targeting, geo-fencing, and Beacons.
\ifCLASSOPTIONcompsoc
  \section*{Acknowledgments}
\else
  \section*{Acknowledgment}
\fi This research was supported by the Basic Science Research
Program through the National Research Foundation of Korea (NRF) funded by the Ministry of Education (2017R1D1A1A09000835). The material in this paper was presented in part at the IEEE/ACM International Conference on Advances in Social Networks Analysis and Mining, Sydney, Australia, July/August 2017, and has been significantly extended based on the prior work~\cite{minh}. Won-Yong Shin is the corresponding author.

\end{document}